%% file: main.tex
\documentclass[11pt]{amsart}

\usepackage[T1]{fontenc}
\usepackage[utf8]{inputenc}
\usepackage{lmodern}
\usepackage{amsmath,amssymb,amsthm,mathtools}
\usepackage{microtype}
\usepackage{booktabs}
\usepackage{array}
\usepackage{enumitem}
\usepackage{xcolor}
\usepackage{hyperref}
\usepackage{url}

\hypersetup{
  colorlinks=true,
  linkcolor=blue!55!black,
  citecolor=blue!55!black,
  urlcolor=blue!55!black,
  pdftitle={Reinhardt's Maximum-Perimeter Polygon Problem for n=16, 32, and 64},
  pdfauthor={Jizhou Guo and Yitao Luo}
}
\setlist{nosep}
\allowdisplaybreaks
\providecommand{\tightlist}{%
  \setlength{\itemsep}{0pt}\setlength{\parskip}{0pt}}

\newtheorem{theorem}{Theorem}[section]
\newtheorem{proposition}[theorem]{Proposition}
\newtheorem{lemma}[theorem]{Lemma}

\theoremstyle{definition}

\theoremstyle{remark}

\newcommand{\norm}[1]{\left\lVert #1\right\rVert}
\newcommand{\abs}[1]{\left\lvert #1\right\rvert}
\newcommand{\perim}{\operatorname{per}}
\newcommand{\diam}{\operatorname{diam}}

\title[Reinhardt's problem for $n=16,32,64$]{Reinhardt's Maximum-Perimeter Polygon Problem for $n=16$, $32$, and $64$}
\author{Jizhou Guo}
\address{Dots Studio, RedNote, Shanghai, China}
\email{mitsuha2021b@gmail.com}
\urladdr{https://aster2024.github.io/}
\thanks{ORCID: \href{https://orcid.org/0009-0001-0699-9164}{0009-0001-0699-9164}.}
\author{Yitao Luo}
\address{University of Science and Technology of China, Hefei, China}
\email{luoyt.kd@mail.ustc.edu.cn}
\date{August 25, 2026}

\begin{document}

\begin{abstract}
A convex polygon is called small if its diameter is at most one. Reinhardt proved the universal perimeter bound
\[
  \perim(P)\le U_n:=2n\sin\frac{\pi}{2n},
\]
and the bound is attained whenever $n$ has a nontrivial odd divisor. The remaining power-of-two cases have resisted exact solution beyond $n=8$. We give computer-assisted proofs of the first three cases, $n=16,32,64$, and in each case prove uniqueness of the maximizing congruence class. The proof architecture is common to all three cases: pass to the difference body $P-P$; encode its reconstruction by a sign code; prove that every global maximizer is saturated, so all difference-body vertices lie on the unit circle; localize every competitive configuration near the regular angle vector; exhaustively screen the sign codes using exact arithmetic; eliminate all nonwinning dihedral orbits; and prove uniqueness inside the winning code by strong convexity and a quantitative KKT argument. The exact certificates cover $2^{15}$ normalized codes for $n=16$, $2^{31}$ normalized codes for $n=32$, and all $2^{64}$ half-codes for $n=64$, leaving respectively $16$, $96$, and $896$ survivors before orbit elimination. The accompanying source package contains the verifiers, recorded outputs, hashes, and separate computational cross-checks.
\end{abstract}

\maketitle

\section{Introduction}

A planar convex $n$-gon $P$ is \emph{small} if $\diam(P)\le 1$. The problem of maximizing $\perim(P)$ under this constraint goes back to Reinhardt \cite{reinhardt1922}. His upper bound
\begin{equation}
  \perim(P)\le U_n:=2n\sin\frac{\pi}{2n}
  \label{eq:reinhardt-bound}
\end{equation}
is sharp whenever $n$ is not a power of two: a nontrivial odd divisor of $n$ yields a Reinhardt polygon attaining equality. The combinatorics and enumeration of such equality cases are by now well developed; see, for example, \cite{haremossinghoff2013,haremossinghoff2019}.

For $n=2^s$ with $s\ge 4$, equality in \eqref{eq:reinhardt-bound} is impossible and the exact optimum has remained unknown. Analytic and numerical constructions produced increasingly sharp lower bounds and high-precision candidates \cite{mossinghoff2006,mossinghoff2008,binganeaudet2022,bingane2023}. Global-optimization formulations sharpened this picture, in some cases under unproved structural conjectures \cite{bingane2022,mulanskypotschka2025}. In particular, the zonogon formulation of Mulansky and Potschka computes highly accurate candidates for $n=16,32,64,128$, while explicitly leaving the global proof problem open \cite{mulanskypotschka2025}. The publisher's correction concerns production typesetting and does not alter those mathematical conclusions \cite{mulanskypotschka2025correction}.

The present work proves the cases $n=16,32,64$ in one common framework. The central difficulty is not finding a strong numerical approximation. It is proving simultaneously that
\begin{enumerate}[label=(\roman*)]
  \item a global maximizer has a fully saturated difference body;
  \item every possible sign code, not merely symmetric codes, is covered;
  \item all but one dihedral code orbit are excluded by certified inequalities; and
  \item the surviving fixed-code nonlinear problem has at most one high-perimeter KKT point.
\end{enumerate}

\paragraph{Contributions.}
The manuscript establishes the following results, with the finite steps backed by the ancillary certificates.
\begin{enumerate}[label=(\alph*)]
  \item A reconstruction lemma turns a centrally symmetric difference body and a closed sign code into a feasible small polygon.
  \item A two-sided convex perturbation shows that at most one half-vertex of the difference body can lie inside the unit disk.
  \item Cauchy's perimeter formula, quantitative near-regularity, and KKT stationarity exclude that last possible interior vertex for each of $n=16,32,64$.
  \item A regular-residual expansion reduces the global code search to exact finite certificates.
  \item For $n=64$, a reflection-pair decomposition converts all $2^{64}$ codes into two complementary ternary subset-sum problems of size $3^{16}$ per block.
  \item A uniform strong-convexity and multiplier argument proves uniqueness in the winning code orbit.
\end{enumerate}

\paragraph{Verification.}
The exact programs have been rerun, and separately organized implementations reproduce the survivor sets and the high-precision stationary points. The source package records the verifier inputs, outputs, manifests, and cross-check implementations used for the three finite certificates.

\section{Main theorems and numerical identification}

Let
\[
  M_n=\sup\{\perim(P): P\text{ is a convex }n\text{-gon and }\diam(P)\le 1\}.
\]
Congruence below allows translation, rotation, and reflection; cyclic relabeling is immaterial.

\begin{theorem}[$n=16$]
The value $M_{16}$ is attained by a unique congruence class. In the saturated difference-body model, its sign code is the dihedral class of
\[
  (+--+-++-)^2.
\]
Numerically,
\[
  M_{16}=3.13654771648660738608596703194\ldots.
\]
\end{theorem}

\begin{theorem}[$n=32$]
The value $M_{32}$ is attained by a unique congruence class. A representative of the winning sign-code orbit is
\begin{center}
\ttfamily\small
+++-+-+-++--+--+--++-+-+-+++--++
\end{center}
and an axial representative is
\begin{center}
\ttfamily\small
+-++--+-+-+---++--+++-+-+-++--+-
\end{center}
Numerically,
\[
  M_{32}=3.1403311569546193658254013805774586723120530983\ldots.
\]
\end{theorem}

\begin{theorem}[$n=64$]
The value $M_{64}$ is attained by a unique congruence class. A representative of the winning saturated half-code is
\begin{center}
\ttfamily\footnotesize
-++++++-----+--+-+++--+---+--+++---++-+++-++---+-++-+++++------+
\end{center}
Numerically,
\[
  M_{64}=3.14127725093277286806199141550246829795626209630809641\ldots.
\]
\end{theorem}

The decimal values identify the maximizers but are not used as floating-point proof decisions. The global certificates use the thresholds summarized in Table~\ref{tab:summary}.

\begin{table}[ht]
\centering
\caption{Summary of the three exact finite certificates.}
\label{tab:summary}
\begin{tabular}{@{}>{\centering\arraybackslash}p{0.08\textwidth}p{0.27\textwidth}p{0.22\textwidth}p{0.18\textwidth}p{0.13\textwidth}@{}}
\toprule
$n$ & Competitive threshold & Code space covered & Survivors & Dihedral orbits \\
\midrule
16 & $\perim(P)>3.1365475$ & $2^{15}$ normalized codes & 16 & 1 \\
32 & $U_{32}-\perim(P)\le1.35\cdot10^{-13}$ & $2^{31}$ normalized codes & 96 & 3 \\
64 & $U_{64}-\perim(P)\le2.84\cdot10^{-23}$ & all $2^{64}$ half-codes & 896 & 6 \\
\bottomrule
\end{tabular}
\end{table}

\section{Difference bodies and sign-code reconstruction}

Let $P$ be a convex small $n$-gon and let
\[
  Z=P-P=P+(-P).
\]
Then $Z$ is centrally symmetric, $Z\subseteq\overline B(0,1)$, and planar perimeter is Minkowski additive, so
\begin{equation}
  \perim(Z)=2\perim(P).
  \label{eq:per-difference}
\end{equation}
The cyclic edge list of $Z$ is obtained by merging the edge lists of $P$ and $-P$, hence $Z$ has at most $2n$ genuine edges.

Suppose $Z$ has exactly $2n$ genuine vertices. Label one half cyclically by
\[
 z_0,z_1,\ldots,z_{n-1},\qquad z_n=-z_0,
\]
and write $e_j=z_{j+1}-z_j$.

\begin{lemma}[Sign-code reconstruction]
For $c=(c_0,\ldots,c_{n-1})\in\{\pm1\}^n$, the condition
\begin{equation}
  \sum_{j=0}^{n-1}c_je_j=0
  \label{eq:edge-closure}
\end{equation}
is necessary and sufficient for the existence of a strictly convex $n$-gon $P_c$, unique up to translation, whose selected edge set is $\{c_je_j\}$ and whose difference body is $Z$.
\end{lemma}

\begin{proof}
Necessity is closure of the edge list. Conversely, the vectors $c_je_j$ are nonzero and have pairwise distinct unoriented directions. Sorting them by polar angle and using \eqref{eq:edge-closure} produces the edge list of a closed strictly convex polygon. The edge list of its negative is the antipodal complement, so merging the two lists gives the full cyclic edge list of $Z$. Convex polygons with the same cyclic edge list differ by translation, and centered difference bodies therefore coincide.
\end{proof}

Summation by parts rewrites \eqref{eq:edge-closure} as
\begin{equation}
 -(c_0+c_{n-1})z_0+\sum_{j=1}^{n-1}(c_{j-1}-c_j)z_j=0.
 \label{eq:vertex-closure}
\end{equation}
All coefficients in \eqref{eq:vertex-closure} belong to $\{-2,0,2\}$. This formulation is crucial because it remains valid under small perturbations of the vertices and therefore links the finite code to the original geometric feasible set.

\begin{lemma}[At most one interior half-vertex]
At a local perimeter maximizer with a genuine $2n$-vertex difference body, at most one of $z_0,\ldots,z_{n-1}$ lies strictly inside the unit disk. If such a vertex exists, its coefficient in \eqref{eq:vertex-closure} is nonzero.
\end{lemma}

\begin{proof}
If two interior vertices $z_r,z_s$ have nonzero coefficients $a_r,a_s$, perturb them by
\[
 z_r(t)=z_r+t a_sh,\qquad z_s(t)=z_s-t a_rh.
\]
If an interior vertex has zero coefficient, move it alone. The closure equality is preserved. For both signs of sufficiently small $t$, disk containment, strict cyclic order, and the genuine-vertex inequalities remain valid. Along this line the perimeter is a sum of norms of affine functions and is therefore convex. Choosing $h$ outside finitely many edge-parallel directions makes at least one summand strictly convex, contradicting a two-sided local maximum.
\end{proof}

\section{Saturation of a global maximizer}

The next step is to prove that every vertex of an extremal difference body lies on the unit circle. This removes radial variables and reduces the problem to angles and a sign code.

First, a certified feasible lower bound is separated from the best possible perimeter of a centrally symmetric polygon with fewer than $2n$ vertices. Since an $m$-gon in the unit disk has perimeter at most $2m\sin(\pi/m)$, this forces the difference body of any global maximizer to have all $2n$ genuine vertices.

For a vertex $z_j$, write $r_j=\norm{z_j}$, let $\omega_j$ be the width of its outward normal cone, let $\eta_j$ be the midpoint direction of that cone, and let $\phi_j$ be the polar direction of $z_j$. Put $\delta_j=\eta_j-\phi_j$. Cauchy's formula gives
\begin{equation}
  \perim(Z)=\sum_{j=0}^{2n-1}2r_j\sin\frac{\omega_j}{2}\cos\delta_j,
  \qquad \sum_{j=0}^{2n-1}\omega_j=2\pi.
  \label{eq:cauchy-cones}
\end{equation}
Jensen's inequality gives the universal upper bound in \eqref{eq:reinhardt-bound}. More importantly, the nonnegative deficit decomposition
\begin{align}
 2U_n-\perim(Z)
 &=\left(2U_n-\sum_j2\sin\frac{\omega_j}{2}\right) \\
 &\quad+\sum_j2\sin\frac{\omega_j}{2}\bigl(1-r_j\cos\delta_j\bigr)
 \label{eq:deficit-decomp}
\end{align}
turns a very small global deficit into pointwise bounds on every $\omega_j$, $r_j$, and $\delta_j$. In particular, consecutive radial directions acquire a uniform positive projective separation.

By the preceding perturbation lemma, there can be at most one interior half-vertex $z_r$. At such a point MFCQ holds: move all active unit-circle vertices radially inward and use the free two-dimensional displacement of $z_r$ to restore \eqref{eq:vertex-closure}. The gradient of $\perim(Z)/2$ with respect to $z_j$ is
\begin{equation}
  g_j=2\sin\frac{\omega_j}{2}e^{i\eta_j}.
  \label{eq:vertex-gradient}
\end{equation}
Stationarity at the interior vertex fixes the projective direction of the equality multiplier to be $\eta_r$. A second nonzero coefficient in \eqref{eq:vertex-closure}, together with tangential stationarity at its active circle vertex, forces its radial direction to lie too close to $\phi_r$, contradicting the projective separation obtained from \eqref{eq:deficit-decomp}. The explicit constants are different in the three cases and are certified in Appendices~\ref{app:n16}--\ref{app:n64}.

\begin{proposition}[Saturation]
For each of $n=16,32,64$, every global maximizer has a difference body with exactly $2n$ genuine vertices, all lying on the unit circle.
\end{proposition}

\section{The saturated angle-code problem}

After rotation, write
\[
  0=\phi_0<\phi_1<\cdots<\phi_n=\pi,
  \qquad \alpha_j=\phi_{j+1}-\phi_j.
\]
Then
\begin{equation}
  \perim(P)=F(\alpha):=\sum_{j=0}^{n-1}2\sin\frac{\alpha_j}{2},
  \qquad \sum_{j=0}^{n-1}\alpha_j=\pi,
  \label{eq:angle-objective}
\end{equation}
and the sign-code closure is
\begin{equation}
  G_c(\phi):=\sum_{j=0}^{n-1}c_j\bigl(e^{i\phi_{j+1}}-e^{i\phi_j}\bigr)=0.
  \label{eq:complex-closure}
\end{equation}
Set $\theta=\pi/n$, $\phi_j=j\theta+s_j$, and $x_j=s_{j+1}-s_j$. A one-gap Jensen envelope localizes every competitive $\alpha_j$ around $\theta$. Strong concavity of \eqref{eq:angle-objective} then gives
\begin{equation}
  U_n-F(\alpha)\ge \kappa_n\norm{x}_2^2,
  \label{eq:strong-concavity}
\end{equation}
with a certified $\kappa_n>0$. Thus every competitive configuration lies in a small Euclidean ball around the regular angle vector.

Let $\xi=e^{i\theta}$. Expanding \eqref{eq:complex-closure} at the regular point gives
\begin{equation}
  G_c=b_c+L_c(s)+R_c(s),
  \qquad
  b_c=(\xi-1)\sum_{j=0}^{n-1}c_j\xi^j.
  \label{eq:residual-expansion}
\end{equation}
The nonlinear remainder is bounded by a Dirichlet Poincare inequality. The following uniform estimate controls the linear term.

\begin{lemma}[Twisted Dirichlet bound]
For every $c\in\{\pm1\}^n$,
\begin{equation}
  \abs{L_c(s)}\le \sqrt{2n}\cos\frac{\pi}{2n}\,\norm{x}_2.
  \label{eq:twisted-bound}
\end{equation}
\end{lemma}

\begin{proof}
Cauchy--Schwarz gives
\[
 \abs{L_c(s)}^2\le n\sum_{j=0}^{n-1}\abs{\xi s_{j+1}-s_j}^2.
\]
On $s_1,\ldots,s_{n-1}$, the twisted Dirichlet form has eigenvalues
\[
 \mu_k=2-2\cos\theta\cos(k\theta),
\]
whereas $\sum_jx_j^2$ has eigenvalues $\lambda_k=2-2\cos(k\theta)$. For $k\ge1$,
\[
 \mu_k\le2\cos^2(\theta/2)\lambda_k,
\]
because half of the difference between the right and left sides is $\cos\theta-\cos(k\theta)\ge0$.
\end{proof}

Equations \eqref{eq:strong-concavity}--\eqref{eq:twisted-bound} imply a necessary upper bound on the regular residual $\abs{b_c}$. This converts a continuous global problem into a finite exact code screen.

\section{Exact finite certificates}

\subsection{$n=16$: complete direct scan}

After normalizing $c_0=+1$, the verifier scans all $2^{15}=32768$ codes. Exact rational and outward integer intervals bound $\abs{b_c}$ and the code-specific linear operator norm
\[
  \sigma_c^2=\lambda_{\max}(B_cD^{-1}B_c^T),
\]
where $D$ is the Dirichlet path matrix. Exactly 16 codes survive the necessary inequality. They form one dihedral orbit, represented by $(+--+-++-)^2$. The smallest exclusion margin among all other codes is greater than $0.002111068044745996$.

\subsection{$n=32$: binary meet-in-the-middle}

There are $2^{31}$ normalized codes. The verifier encloses each regular edge $e_j=\xi^{j+1}-\xi^j$ by rational intervals, rounds it to a fixed-point integer vector at scale $10^{18}$ with a proved coordinate error, and splits the sign sum into blocks of sizes 16 and 16. An exact range search returns 96 codes. They form three dihedral orbits, represented by
\begin{center}
\ttfamily\small
A = ++++--+--+-+-+-+++---+-+-+-++-++\\
B = +++-+-+-++--+--+--++-+-+-+++--++\\
C = ++-++-+-+-+-+--++--+-+-+-+-++-++.
\end{center}
Code-specific exact bounds on $\abs{b_c}$ and $\sigma_c$ show that the best possible deficits in orbits $A$ and $C$ exceed the certified feasible deficit. Only orbit $B$ remains.

\subsection{$n=64$: reflection-pair ternary decomposition}

A direct normalized scan would involve $2^{63}$ codes. Pair the indices $j$ and $63-j$ and set
\[
 \beta_j=\left(j-\frac{63}{2}\right)\theta,
 \qquad
 p_j=\frac{c_j+c_{63-j}}2,
 \qquad
 q_j=\frac{c_j-c_{63-j}}2.
\]
Exactly one of $p_j,q_j$ is zero, and the other is $\pm1$. Direct calculation gives
\begin{equation}
 \sum_{j=0}^{63}c_j\xi^j
 =2e^{63i\theta/2}
 \left(\sum_{j=0}^{31}p_j\cos\beta_j
 +i\sum_{j=0}^{31}q_j\sin\beta_j\right).
 \label{eq:reflection-decomp}
\end{equation}
Conversely, complementary ternary supports in the two coordinates determine one and only one binary code. Thus \eqref{eq:reflection-decomp} is a bijective reparameterization of all $4^{32}=2^{64}$ half-codes, not a symmetry assumption.

Each ternary coordinate sum is split into two 16-term blocks, producing lists of $3^{16}=43{,}046{,}721$ states. Fixed-point interval guarantees ensure that every code satisfying the analytic residual threshold is included. The exact scan leaves 896 codes in six dihedral orbits of sizes
\[
 128,128,128,128,128,256.
\]
Code-specific lower residual bounds and upper operator-norm bounds eliminate five orbits; the smallest certified deficit excess is $8.50\cdot10^{-24}$.

\section{Uniqueness inside the winning code}

Fix a winning code and use the internal angles $u=(\phi_1,\ldots,\phi_{n-1})$. Let $f=-F$ and let $g=(\Re G_c,\Im G_c)$. In the certified high-perimeter region, the Hessian of $f$ is a weighted Dirichlet path matrix with positive weights, hence
\begin{equation}
  \nabla^2 f\succeq mI
  \label{eq:strong-convexity-fixed}
\end{equation}
for an explicit $m>0$.

The nonzero columns of $Dg$ occur at the switch indices $J=\{j:c_{j-1}\ne c_j\}$ and are $a_jie^{i\phi_j}$ with $a_j=\pm2$. Therefore
\begin{equation}
  \sigma_{\min}(Dg)^2
  =2\left(\abs{J}-\left|\sum_{j\in J}e^{2i\phi_j}\right|\right).
  \label{eq:jacobian-singular}
\end{equation}
An exact root-of-unity bound at the regular point, plus the near-regular radius, yields a uniform lower bound on $\sigma_{\min}(Dg)$. At a KKT point $\nabla f+Dg^Ty=0$, this gives a small multiplier bound $\norm{y}\le Y$.

\begin{proposition}[Quantitative KKT uniqueness]
If the certified constants satisfy $m>2Y$, then the fixed code has at most one feasible KKT point in the high-perimeter region after fixing rotation.
\end{proposition}

\begin{proof}
Let $(u,y_u)$ and $(v,y_v)$ be two feasible KKT pairs and put $d=u-v$. Strong convexity gives
\[
 d^T(\nabla f(u)-\nabla f(v))\ge m\norm{d}^2.
\]
The exponential Taylor remainder and $\abs{a_j}\le2$ give
\[
 \norm{g(u)-g(v)-Dg(v)d}\le\norm{d}^2.
\]
The analogous reverse estimate also holds. Subtracting the stationarity equations, taking the inner product with $d$, and using feasibility yields
\[
 m\norm{d}^2\le(\norm{y_u}+\norm{y_v})\norm{d}^2\le2Y\norm{d}^2.
\]
Thus $d=0$ when $m>2Y$.
\end{proof}

The certified margins are summarized in Table~\ref{tab:kkt}.

\begin{table}[ht]
\centering
\caption{Certified uniqueness margins in the winning fixed-code problem.}
\label{tab:kkt}
\begin{tabular}{@{}ccc@{}}
\toprule
$n$ & switch indices & certified inequality \\
\midrule
16 & 11 & $m>0.0018>0.000302>2Y$ \\
32 & 21 & $m-2Y>8.7799\cdot10^{-5}$ \\
64 & 27 & $m-2Y>2.22\cdot10^{-5}$ \\
\bottomrule
\end{tabular}
\end{table}

Compactness supplies a global maximizer. Saturation puts it in the angle-code model, the finite certificate puts it in the winning orbit, and strict angle inequalities plus the Jacobian bound make it a KKT point. Quantitative uniqueness then leaves exactly one normalized optimizer; dihedral changes, rotation, reflection, and translation produce the same congruence class.

\section{Verification artifacts and reproducibility}

The evolving source is maintained in a
\href{https://github.com/aster2024/reinhardt-powers-of-two-proof-candidates}{public GitHub repository}.
Archived releases are collected in the
\href{https://doi.org/10.5281/zenodo.21796494}{Zenodo concept record};
the snapshot available when this revision was prepared is the
\href{https://doi.org/10.5281/zenodo.21796495}{version-specific Zenodo record}.
The version DOI identifies the immutable certificate snapshot used here, while the arXiv ancillary package carries the same hashed verification artifacts.

The arXiv source package is designed to include the exact programs as ancillary files under \texttt{anc/}. The principal artifacts are:
\begin{itemize}
  \item \texttt{anc/n16/reinhardt\_n16\_verified\_audit.py};
  \item \texttt{anc/n32/reinhardt\_n32\_verifier.py};
  \item \texttt{anc/n64/n64\_analytic\_verifier.py};
  \item \texttt{anc/n64/n64\_code\_fixed.cpp};
  \item \texttt{anc/n64/n64\_post\_verifier.py};
  \item recorded outputs and the $n=64$ survivor file;
  \item separately organized scans and high-precision KKT cross-checks under \texttt{anc/crosschecks/}.
\end{itemize}

The $n=16$ and $n=32$ core verifiers use the Python standard library and exact integer or rational arithmetic. The $n=64$ analytic verifier now also certifies all 64 fixed-point weights embedded in the C++ scan and the endpoint Jensen inequalities used for localization. The post-verifier uses a 256-unit residual padding, which dominates the generic two-coordinate rounding error $128\sqrt2$. The $n=64$ exhaustive scan is C++17 and is memory intensive; the recorded run used approximately 1.47 GB. No binary floating-point comparison is used to accept or exclude a code in the proof certificates. High-precision floating calculations are used only to identify the final stationary points and to cross-check the reported decimal values.

A minimal reproduction sequence from the source-package root is:
\begin{verbatim}
bash anc/run_core_verifiers.sh
\end{verbatim}
The runner changes into each case directory before execution and regenerates the $n=64$ survivor file before invoking the post-verifier. The ancillary README also gives the directory-local commands.

\section{Computational assistance}

OpenAI Codex was used as an assistant for mathematical exploration, proof checking, and manuscript preparation.

\section{Conclusion}

The three cases share a coherent proof mechanism: a geometric saturation argument reduces the original problem to a near-regular sign-coded angle problem; exact finite certificates isolate a single code orbit; and quantitative nonlinear optimization proves uniqueness inside that orbit. The most technically novel computational ingredient is the $n=64$ reflection-pair ternary decomposition, which covers all $2^{64}$ half-codes without assuming axial symmetry. Together, the analytic reductions and exact certificates resolve the first three open power-of-two cases of Reinhardt's perimeter problem.

\appendix

\section{Detailed certificate narrative for $n=16$}
\label{app:n16}
\input{appendix_n16}

\section{Detailed certificate narrative for $n=32$}
\label{app:n32}
\input{appendix_n32}

\section{Detailed certificate narrative for $n=64$}
\label{app:n64}
\input{appendix_n64}

\end{document}

%% file: appendix_n16.tex
\subsection{Statement}

A \textbf{small polygon} is a planar convex polygon of diameter at most one. Write \(p(P)\) for its perimeter.

\subsubsection{Main theorem}

Among all convex small hexadecagons, the maximum perimeter is attained by a unique congruence class (allowing translation, rotation, reflection, and cyclic relabeling). In the saturated difference-body representation its sign code is the dihedral class of

\[
 c_*=(+--+-++-)^2.
\]

Its perimeter is numerically

\[
 3.136547716486607386085967\ldots .
\]

The decimal value is included only to identify the optimizer; the global proof below does not depend on a decimal interval for the final stationary point.

\begin{center}\rule{0.5\linewidth}{0.5pt}\end{center}

\subsection{Difference bodies and the first strict reduction}

Let \(P\) be a convex small hexadecagon and put

\[
 Z=P-P=P+(-P).
\]

Then \(Z\) is centrally symmetric, \(Z\subseteq \overline B(0,1)\), and Minkowski additivity of planar perimeter gives

\[
 p(Z)=2p(P). \tag{2.1}
\]

The edge directions of a Minkowski sum are obtained by merging the edge-direction lists of its summands. Consequently, \(Z\) has at most \(32\) genuine edges.

We use the following certified lower bound, supplied by Bingane's published explicit feasible construction \(C_{16}\) \cite{bingane2023}:

\[
 p(C_{16})>L_0:=3.1365475. \tag{2.2}
\]

Feasibility of \(C_{16}\) is imported from the cited construction; the exact verifier independently encloses its published radical/trigonometric perimeter formula by nested radical intervals and proves (2.2). Thus this step does not claim an independent coordinate-level feasibility certificate.

\subsubsection{Lemma 2.1 --- the extremal difference body has 32 strict vertices}

Let \(P\) maximize perimeter. Then \(Z=P-P\) has exactly \(32\) genuine vertices and edges.

\paragraph{Proof}

If a centrally symmetric polygon has fewer than \(32\) vertices, it has at most \(30\). An \(m\)-gon contained in the unit disk has perimeter at most

\[
 2m\sin\frac{\pi}{m},
\]

with equality only for the regular inscribed \(m\)-gon. Since \(m\sin(\pi/m)\) increases with \(m\), (2.1) yields

\[
 p(P)\le 30\sin\frac{\pi}{30}<3.1365475,
\]

where the last inequality is exactly certified. This contradicts (2.2). \hfill\(\square\)

In particular, an extremal \(P\) has sixteen nonzero edges with pairwise distinct unoriented directions; otherwise the merged edge list of \(P+(-P)\) would have fewer than \(32\) members.

\begin{center}\rule{0.5\linewidth}{0.5pt}\end{center}

\subsection{Exact reconstruction from a sign code}

Let \(Z\) be a strictly convex centrally symmetric \(32\)-gon. Label one half of its vertices in counterclockwise order by

\[
 z_0,z_1,\ldots,z_{15},\qquad z_{16}=-z_0,
\]

and define its half-edge vectors

\[
 e_j=z_{j+1}-z_j,\qquad 0\le j\le15.
\]

The other sixteen edge vectors are \(-e_0,\ldots,-e_{15}\).

\subsubsection{Lemma 3.1 --- reconstruction lemma}

For \(c=(c_0,\ldots,c_{15})\in\{\pm1\}^{16}\), the condition

\[
 \sum_{j=0}^{15}c_je_j=0 \tag{3.1}
\]

is necessary and sufficient for the existence of a strictly convex hexadecagon \(P_c\) whose difference body is \(Z\) and whose edge set is

\[
 \{c_0e_0,\ldots,c_{15}e_{15}\}.
\]

The polygon \(P_c\) is unique up to translation.

\paragraph{Proof}

Necessity follows from closure of the oriented edge list of \(P_c\).

Conversely, put \(f_j=c_je_j\). Because \(Z\) has \(32\) strict edges, the \(f_j\) are nonzero and have pairwise distinct directions. Sort them by polar angle. Equation (3.1) says that the sorted list closes. A cyclic list of nonzero vectors with strictly increasing directions and zero sum is the edge list of a strictly convex polygon, unique up to translation.

The edge list of \(-P_c\) is \(-f_0,\ldots,-f_{15}\). Hence the edge-direction merge for \(P_c+(-P_c)\) is exactly

\[
 \{f_j,-f_j:0\le j\le15\}
 =\{e_j,-e_j:0\le j\le15\},
\]

which is the full cyclic edge list of \(Z\). Two convex polygons with the same cyclic edge list differ by a translation. Both difference bodies are centrally symmetric about the origin, so the translation is zero and \(P_c-P_c=Z\). \hfill\(\square\)

Expanding (3.1) by summation by parts gives the vertex form

\[
 \sum_{j=0}^{15}a_jz_j=0, \tag{3.2}
\]

where

\[
 a_0=-(c_0+c_{15}),\qquad
 a_j=c_{j-1}-c_j\quad(1\le j\le15). \tag{3.3}
\]

Thus every \(a_j\in\{-2,0,2\}\).

Lemma 3.1 is also the needed local-feasibility statement: any sufficiently small perturbation of the \(z_j\) that preserves central symmetry, strict cyclic order, disk containment, and (3.2) reconstructs a feasible small hexadecagon with the same code.

\begin{center}\rule{0.5\linewidth}{0.5pt}\end{center}

\subsection{At most one interior vertex in a half difference body}

For the half-edge representation, (2.1) gives

\[
 p(P)=\sum_{j=0}^{15}\lVert z_{j+1}-z_j\rVert. \tag{4.1}
\]

\subsubsection{Lemma 4.1 --- two interior half-vertices are impossible}

At a local maximizer satisfying Lemma 2.1, at most one of \(z_0,\ldots,z_{15}\) lies strictly inside the unit disk. If \(z_r\) is the unique interior half-vertex, then \(a_r\ne0\).

\paragraph{Proof}

Suppose first that \(z_r,z_s\) are two distinct interior half-vertices.

If \(a_r=0\), set

\[
 z_r(t)=z_r+th
\]

and leave the other half-vertices fixed. If both coefficients are nonzero, set

\[
 z_r(t)=z_r+t a_sh,\qquad
 z_s(t)=z_s-t a_rh. \tag{4.2}
\]

If exactly one coefficient vanishes, vary the corresponding zero-coefficient vertex alone. In every case (3.2) is preserved.

Because the moved vertices are strictly inside the disk, both signs of sufficiently small \(t\) preserve disk containment. Strict convexity, nonzero edges, and cyclic vertex order are open conditions, so they too persist for \(|t|\) small. Lemma 3.1 therefore reconstructs a feasible polygon along the whole small two-sided interval.

Along this line, (4.1) is a sum of functions of the form

\[
 t\longmapsto\lVert u+tv\rVert,
\]

hence is convex. At least one edge changes. Choose \(h\) outside the finite set of directions parallel to the affected edges. For that edge, \(u\) and \(v\) are not parallel, so its norm is strictly convex. Thus the whole perimeter is strictly convex in \(t\). Strict convexity gives

\[
 p(P_0)<\max\{p(P_{-t}),p(P_t)\}
\]

for every sufficiently small \(t>0\), contradicting local maximality.

If a unique interior vertex \(z_r\) had \(a_r=0\), the same one-vertex perturbation would give the same contradiction. \hfill\(\square\)

This proof includes the adjacent-vertex and endpoint cases: even if the variation of their common edge cancels, at least one outer incident edge has nonzero variation.

\begin{center}\rule{0.5\linewidth}{0.5pt}\end{center}

\subsection{Quantitative near-regularity of every competitor}

Let the full cyclic vertex list of \(Z\) be indexed modulo \(32\). For vertex \(z_j\), write

\[
 r_j=\lVert z_j\rVert,
\]

let \(\omega_j\in(0,\pi)\) be the width of its outward normal cone, let \(\eta_j\) be the midpoint direction of that cone, let \(\phi_j\) be the polar direction of \(z_j\), and choose

\[
 \delta_j=\eta_j-\phi_j\in[-\pi,\pi].
\]

Cauchy's perimeter formula, integrated separately over the normal cones, gives the exact identity

\[
 p(Z)=\sum_{j=0}^{31}2r_j\sin\frac{\omega_j}{2}\cos\delta_j,
 \qquad
 \sum_{j=0}^{31}\omega_j=2\pi. \tag{5.1}
\]

The universal inscribed-polygon bound from (5.1) and Jensen is

\[
 p(Z)\le U_Z:=64\sin\frac{\pi}{32}. \tag{5.2}
\]

For an optimizer, (2.2) and (2.1) imply

\[
 0\le U_Z-p(Z)<1.982\times10^{-6}. \tag{5.3}
\]

All numerical inequalities in the next lemma are exact verifier assertions.

\subsubsection{Lemma 5.1 --- normal-cone, radial, and angular localization}

Every vertex of an extremal \(Z\) satisfies

\[
 0.18<\omega_j<0.21, \tag{5.4}
\]

\[
 r_j>0.99998, \tag{5.5}
\]

\[
 |\delta_j|<0.005. \tag{5.6}
\]

Moreover, the projective angular distance between any two distinct half-vertices is greater than \(0.17\).

\paragraph{Proof}

Fix \(\omega_j=t\). Dropping the factors \(r_k\cos\delta_k\le1\) and applying Jensen to the other 31 cone widths gives

\[
 p(Z)\le F_\omega(t)
 :=2\sin\frac t2+62\sin\frac{2\pi-t}{62}. \tag{5.7}
\]

The function increases up to \(2\pi/32\) and decreases afterwards. The verifier proves

\[
 F_\omega(0.18)<2L_0,
 \qquad
 F_\omega(0.21)<2L_0.
\]

Since \(p(Z)>2L_0\), (5.4) follows.

Rewrite the deficit as

\[
\begin{aligned}
 U_Z-p(Z)
 &=\left(U_Z-\sum_j2\sin\frac{\omega_j}{2}\right)\\
 &\quad+\sum_j2\sin\frac{\omega_j}{2}
       \left(1-r_j\cos\delta_j\right).
\end{aligned} \tag{5.8}
\]

Every term on the right is nonnegative. Hence

\[
 U_Z-p(Z)
 \ge2(1-r_j)\sin\frac{\omega_j}{2}.
\]

Equations (5.3), (5.4), and the certified inequality

\[
 2(1-0.99998)\sin(0.09)>1.982\times10^{-6}
\]

give (5.5).

Likewise, from

\[
 1-r_j\cos\delta_j
 =(1-r_j)+r_j(1-\cos\delta_j)
 \ge r_j(1-\cos\delta_j),
\]

we obtain

\[
 U_Z-p(Z)
 \ge2r_j\sin\frac{\omega_j}{2}(1-\cos\delta_j).
\]

The verifier checks that the right side at \(r_j=0.99998\), \(\omega_j=0.18\), and \(|\delta_j|=0.005\) exceeds the deficit bound (5.3). This proves (5.6).

If the consecutive normal-cone midpoint directions are lifted cyclically, then

\[
 \eta_{j+1}-\eta_j
 =\frac{\omega_j+\omega_{j+1}}2>0.18. \tag{5.9}
\]

Using (5.6), every consecutive radial angle gap is therefore greater than

\[
 0.18-2(0.005)=0.17. \tag{5.10}
\]

For two distinct half-vertices, their projective distance is a sum of one or more consecutive full-circle gaps, or the complementary sum to \(\pi\); either sum contains at least one gap from (5.10). \hfill\(\square\)

\begin{center}\rule{0.5\linewidth}{0.5pt}\end{center}

\subsection{Exclusion of the final unsaturated radius}

\subsubsection{Lemma 6.1 --- every difference-body vertex is on the unit circle}

For an extremal hexadecagon,

\[
 \lVert z_j\rVert=1\qquad(0\le j\le31). \tag{6.1}
\]

\paragraph{Proof}

By Lemma 4.1, it is enough to exclude a unique interior half-vertex \(z_r\). In this case \(a_r\ne0\), hence \(|a_r|=2\).

Consider the finite-dimensional program in variables \(z_0,\ldots,z_{15}\): maximize (4.1), subject to the two equality components of (3.2) and the disk inequalities

\[
 q_j(z)=\lVert z_j\rVert^2-1\le0.
\]

All strict convexity and ordering requirements hold on an open neighborhood and introduce no active constraints.

We first verify MFCQ. For every active vertex \(j\ne r\), choose \(d_j=-z_j\). Then

\[
 Dq_j(z)d_j=-2<0.
\]

Because \(a_r\ne0\), choose \(d_r\) uniquely so that

\[
 \sum_ja_jd_j=0.
\]

The equality Jacobian has rank two, again because its \(z_r\)-block is \(a_rI_2\). Thus MFCQ holds and the KKT equations are necessary.

Let \(g_j\) denote the gradient of (4.1) with respect to \(z_j\). The difference of the two unit incident edge tangents is the outward angle-bisector vector, so

\[
 g_j=2\sin\frac{\omega_j}{2}\,e^{i\eta_j}. \tag{6.2}
\]

At the interior vertex the disk multiplier vanishes. Therefore, for some equality multiplier \(\lambda\in\mathbb R^2\),

\[
 g_r=a_r\lambda. \tag{6.3}
\]

Consequently the projective direction of \(\lambda\) is \(\eta_r\), and

\[
 \lVert\lambda\rVert=\sin\frac{\omega_r}{2}. \tag{6.4}
\]

There is another index \(j\ne r\) with \(a_j\ne0\). Otherwise (3.2) would imply \(a_rz_r=0\), forcing \(z_r=0\); but the origin lies in the interior of the two-dimensional centrally symmetric polygon \(Z\), so it is not a vertex.

At this second index \(z_j\) lies on the unit circle. Project its KKT equation onto the tangent direction \(ie^{i\phi_j}\); the radial disk multiplier disappears. Using \(|a_j|=2\), (6.2), and (6.4), we get

\[
 \left|\sin(\eta_r-\phi_j)\right|
 =\frac{\sin(\omega_j/2)}{\sin(\omega_r/2)}
   |\sin\delta_j|. \tag{6.5}
\]

By Lemma 5.1 and exact interval evaluation,

\[
 \frac{\sin(0.105)}{\sin(0.09)}\sin(0.005)<0.00585. \tag{6.6}
\]

If \(d_{\mathbb{RP}^1}\) denotes projective angular distance, then for \(0\le d\le\pi/2\), \(\sin d\ge2d/\pi\). Thus (6.5)--(6.6) imply

\[
 d_{\mathbb{RP}^1}(\eta_r,\phi_j)
 <\frac\pi2(0.00585).
\]

Together with \(|\eta_r-\phi_r|<0.005\), the verifier certifies

\[
 d_{\mathbb{RP}^1}(\phi_r,\phi_j)<0.0142<0.17,
\]

contradicting Lemma 5.1. Hence no interior vertex exists. \hfill\(\square\)

This proves, for \(n=16\), the full-radius saturation property that the prior zonogon computation had assumed conjecturally.

\begin{center}\rule{0.5\linewidth}{0.5pt}\end{center}

\subsection{The saturated angle model}

Rotate so that

\[
 z_0=1,\qquad z_{16}=-1,
\]

and write

\[
 z_j=e^{i\phi_j},\qquad
 0=\phi_0<\phi_1<\cdots<\phi_{16}=\pi.
\]

Set

\[
 \alpha_j=\phi_{j+1}-\phi_j,
 \qquad
 \sum_{j=0}^{15}\alpha_j=\pi. \tag{7.1}
\]

Then

\[
 p(P)=\sum_{j=0}^{15}2\sin\frac{\alpha_j}{2}, \tag{7.2}
\]

and the code closure is

\[
 G_c(\phi):=
 \sum_{j=0}^{15}c_j
 \left(e^{i\phi_{j+1}}-e^{i\phi_j}\right)=0. \tag{7.3}
\]

Let

\[
 \alpha_0^*=\frac\pi{16},\qquad
 x_j=\alpha_j-\alpha_0^*,
 \qquad \sum_jx_j=0. \tag{7.4}
\]

\subsubsection{Lemma 7.1 --- every global competitor is in a tiny regular neighborhood}

Every saturated feasible configuration with perimeter greater than \(L_0\) satisfies

\[
 0.189<\alpha_j<0.204 \tag{7.5}
\]

and

\[
 \lVert x\rVert_2<0.0065. \tag{7.6}
\]

\paragraph{Proof}

If one gap is fixed at \(t\), Jensen on the remaining fifteen gaps gives

\[
 p(P)\le F_\alpha(t)
 :=2\sin\frac t2+30\sin\frac{\pi-t}{30}. \tag{7.7}
\]

This function increases up to \(\pi/16\) and decreases afterwards. The verifier proves

\[
 F_\alpha(0.189)<L_0,
 \qquad
 F_\alpha(0.204)<L_0,
\]

which gives (7.5).

On this interval, for \(h(t)=2\sin(t/2)\),

\[
 h''(t)=-\frac12\sin\frac t2
 \le-\frac12\sin(0.189/2).
\]

Using \(\sum x_j=0\), strong concavity yields

\[
 32\sin\frac\pi{32}-p(P)
 \ge\frac14\sin(0.189/2)\lVert x\rVert_2^2. \tag{7.8}
\]

The verifier proves

\[
 32\sin\frac\pi{32}-L_0<9.91\times10^{-7}.
\]

and, from (7.8), the strict bound (7.6). \hfill\(\square\)

\begin{center}\rule{0.5\linewidth}{0.5pt}\end{center}

\subsection{Exact exclusion of all but one code class}

Define cumulative angle perturbations

\[
 s_0=s_{16}=0,\qquad
 s_j=\sum_{k=0}^{j-1}x_k\quad(1\le j\le15), \tag{8.1}
\]

so that

\[
 \phi_j=\frac{j\pi}{16}+s_j.
\]

The discrete Dirichlet inequality is

\[
 \sum_{j=1}^{15}s_j^2
 \le\frac{1}{4\sin^2(\pi/32)}\sum_{j=0}^{15}x_j^2
 <26.1\lVert x\rVert_2^2. \tag{8.2}
\]

Let \(\xi=e^{i\pi/16}\). Taylor expansion of (7.3) at the regular point gives

\[
 G_c(x)=b_c+L_c(s)+R_c(s), \tag{8.3}
\]

where

\[
 b_c=\sum_{j=0}^{15}c_j(\xi^{j+1}-\xi^j), \tag{8.4}
\]

\[
 L_c(s)=i\sum_{j=1}^{15}(c_{j-1}-c_j)\xi^js_j, \tag{8.5}
\]

and, since \(|e^{it}-1-it|\le t^2/2\),

\[
 |R_c(s)|\le\sum_{j=1}^{15}s_j^2
 <26.1\lVert x\rVert_2^2. \tag{8.6}
\]

Let \(D\) be the \(15\times15\) Dirichlet path Laplacian, so that \(\lVert x\rVert_2^2=s^TDs\). Write \(B_c\) for the real \(2\times15\) matrix representing (8.5), and define

\[
 \sigma_c^2=\lambda_{\max}(B_cD^{-1}B_c^T). \tag{8.7}
\]

Then \(|L_c(s)|\le\sigma_c\lVert x\rVert_2\). Consequently every code admitting a competitor with perimeter greater than \(L_0\) must satisfy

\[
 |b_c|
 \le0.0065\,\sigma_c+26.1(0.0065)^2. \tag{8.8}
\]

\subsubsection{Lemma 8.1 --- finite exact code certificate}

After normalizing \(c_0=+1\), exactly sixteen of the \(2^{15}=32768\) codes satisfy the necessary condition (8.8). They form precisely one dihedral orbit, represented by

\[
 c_*=(+--+-++-)^2. \tag{8.9}
\]

Every other code violates (8.8). The smallest certified violation margin is greater than

\[
 0.002111068044745996. \tag{8.10}
\]

\paragraph{Proof}

This is the finite computer-assisted part. The accompanying verifier:

\begin{enumerate}
\def\labelenumi{\arabic{enumi}.}
\tightlist
\item
  encloses \(\pi\) by exact rational Machin-series bounds;
\item
  constructs \(\sin(\pi/32)\), \(\cos(k\pi/16)\), and \(\sin(k\pi/16)\) using nested square-root integer intervals at scale \(10^{75}\);
\item
  uses the exact formula \[
  (D^{-1})_{jk}=\frac{\min(j,k)(16-\max(j,k))}{16};
  \]
\item
  computes a lower interval for \(|b_c|\) and an upper interval for \(\sigma_c\);
\item
  makes every keep/exclude decision using integers only;
\item
  scans all \(32768\) normalized codes and checks equality of the survivor set with the explicitly generated dihedral orbit of (8.9).
\end{enumerate}

A fresh rerun reproduces the recorded output byte-for-byte and ends with \texttt{ALL\ ASSERTIONS\ PASSED}. \hfill\(\square\)

\begin{center}\rule{0.5\linewidth}{0.5pt}\end{center}

\subsection{Global uniqueness inside the surviving code}

Fix the representative (8.9). Use variables

\[
 u=(\phi_1,\ldots,\phi_{15})
\]

with \(\phi_0=0\), \(\phi_{16}=\pi\), and minimize the negative perimeter

\[
 f(u)=-\sum_{j=0}^{15}2\sin\frac{\phi_{j+1}-\phi_j}{2}. \tag{9.1}
\]

The two real closure constraints are denoted by \(g(u)=0\). The nonzero coefficients \(a_j=c_{j-1}-c_j\) occur at

\[
 S=\{1,3,4,5,7,8,9,11,12,13,15\},
 \qquad |S|=11. \tag{9.2}
\]

\subsubsection{Lemma 9.1 --- uniform strong convexity}

Throughout the high-perimeter box (7.5),

\[
 \nabla^2f(u)\succeq mI,
 \qquad m>0.0018. \tag{9.3}
\]

\paragraph{Proof}

The Hessian is the weighted Dirichlet path Laplacian whose edge weights are

\[
 \frac12\sin\frac{\alpha_j}{2}.
\]

The smallest eigenvalue of the unweighted Dirichlet path Laplacian is \(4\sin^2(\pi/32)\). Hence

\[
 m\ge2\sin(0.189/2)\sin^2(\pi/32)>0.0018,
\]

with the last inequality certified exactly. \hfill\(\square\)

\subsubsection{Lemma 9.2 --- uniform constraint qualification and multiplier bound}

Throughout the same box,

\[
 \sigma_{\min}(Dg(u))>4.42, \tag{9.4}
\]

and every KKT multiplier \(y\in\mathbb R^2\) satisfies

\[
 \lVert y\rVert_2<0.000151. \tag{9.5}
\]

\paragraph{Proof}

The nonzero columns of \(J=Dg(u)\) are

\[
 a_j(-\sin\phi_j,\cos\phi_j)^T,
 \qquad j\in S,
\]

with \(|a_j|=2\). Therefore

\[
 \lambda_{\min}(JJ^T)
 =2\left(11-\left|\sum_{j\in S}e^{2i\phi_j}\right|\right). \tag{9.6}
\]

At the regular point,

\[
 \sum_{j\in S}e^{2ij\pi/16}=-1. \tag{9.7}
\]

This follows directly by subtracting the complementary switch set from the full set of nontrivial sixteenth roots.

With \(\phi_j=j\pi/16+s_j\),

\[
 \left|\sum_{j\in S}e^{2i\phi_j}\right|
 \le1+2\sum_{j\in S}|s_j|
 \le1+2\sqrt{11}\lVert s\rVert_2. \tag{9.8}
\]

By (8.2) and (7.6),

\[
 \lVert s\rVert_2
 <\frac{0.0065}{2\sin(\pi/32)}<0.0332. \tag{9.9}
\]

The exact verifier combines (9.6)--(9.9) to prove (9.4).

For the gradient, let

\[
 q_j=\cos(\alpha_j/2)-\cos(\pi/32).
\]

The internal gradient is a first difference of the \(q_j\), whose operator norm is at most two. The mean-value theorem and (7.5) give

\[
 \lVert\nabla f(u)\rVert_2
 <\sin(0.102)(0.0065)<0.000663. \tag{9.10}
\]

At a KKT point, \(\nabla f+J^Ty=0\), so

\[
 \lVert y\rVert_2
 \le\frac{\lVert\nabla f\rVert_2}{\sigma_{\min}(J)}
 <\frac{0.000663}{4.42}<0.000151.
\] \hfill\(\square\)

\subsubsection{Lemma 9.3 --- at most one high-perimeter KKT point}

The surviving code has at most one KKT point satisfying (7.5)--(7.6).

\paragraph{Proof}

Suppose \((u,y_u)\) and \((v,y_v)\) are two feasible KKT pairs, and put \(d=u-v\). The line segment between \(u\) and \(v\) remains in the angle box, so Lemma 9.1 gives

\[
 d^T(\nabla f(u)-\nabla f(v))\ge m\lVert d\rVert_2^2. \tag{9.11}
\]

For the closure map, its second directional derivative has norm at most

\[
 2\sum_{j=1}^{15}d_j^2=2\lVert d\rVert_2^2,
\]

because \(|a_j|\le2\). Taylor's theorem, including its factor \(1/2\), therefore gives

\[
 \lVert Dg(u)d\rVert_2\le\lVert d\rVert_2^2,
 \qquad
 \lVert Dg(v)d\rVert_2\le\lVert d\rVert_2^2, \tag{9.12}
\]

where feasibility \(g(u)=g(v)=0\) was used.

Subtract the two stationarity equations and take the inner product with \(d\). Using (9.12),

\[
\begin{aligned}
 m\lVert d\rVert_2^2
 &\le d^T(\nabla f(u)-\nabla f(v))\\
 &\le(\lVert y_u\rVert_2+\lVert y_v\rVert_2)
      \lVert d\rVert_2^2\\
 &<0.000302\lVert d\rVert_2^2.
\end{aligned} \tag{9.13}
\]

But \(m>0.0018>0.000302\). Hence \(d=0\). \hfill\(\square\)

\begin{center}\rule{0.5\linewidth}{0.5pt}\end{center}

\subsection{Proof of the main theorem}

The feasible set of convex hexadecagons of diameter at most one, modulo translation and allowing repeated limiting vertices, is compact, and perimeter is continuous. Hence a global maximizer exists. Bingane's construction makes its perimeter greater than \(L_0\).

Lemma 2.1 gives a strict \(32\)-vertex difference body. Lemma 4.1 and Lemma 6.1 show that every difference-body vertex is on the unit circle. Therefore the maximizer is represented by the saturated angle-code model of Section 7.

Lemma 7.1 puts it in the certified regular neighborhood. Lemma 8.1 forces its code into the single dihedral class of \(c_*\). The angle inequalities are strict, and Lemma 9.2 gives full row rank of the closure Jacobian, so the maximizer is a KKT point of the fixed-code problem. Lemma 9.3 says there is at most one such point in the entire region where a global optimizer can lie.

Existence of the original global maximizer supplies at least one such point. Therefore there is exactly one normalized fixed-code optimizer. The sixteen surviving sign strings are merely cyclic/reversal/sign representatives of the same difference body and reconstructed polygon; Lemma 3.1 removes translation ambiguity. This proves uniqueness up to congruence and relabeling. \hfill\(\square\)

\begin{center}\rule{0.5\linewidth}{0.5pt}\end{center}

\subsection{Numerical identification (not used for global validity)}

A high-precision Newton solve for the unique point gives the half-circle gaps

\[
\begin{array}{rcl}
0.19831631349051794937,&&
0.19450334649319957491,\\
0.19450334649319957491,&&
0.19774551481013478084,\\
0.19499397164222383569,&&
0.19716378410929702498,\\
0.19716378410929702498,&&
0.19640626564702685354,\\
0.19640626564702685354,&&
0.19716378410929702498,\\
0.19716378410929702498,&&
0.19499397164222383569,\\
0.19774551481013478084,&&
0.19450334649319957491,\\
0.19450334649319957491,&&
0.19831631349051794937.
\end{array}
\]

They sum to \(\pi\), satisfy the closure equations for \(c_*\), and yield

\[
 p_{16}=3.13654771648660738608596703194\ldots .
\]

For a publication-quality numerical theorem, these coordinates should still be enclosed by an independent interval-Newton computation. That is not needed for the existence and global-uniqueness argument above, because existence comes from compactness and uniqueness from Lemma 9.3.

\begin{center}\rule{0.5\linewidth}{0.5pt}\end{center}

\subsection{Certificate audit and cross-checks}

The following previously dangerous points have been explicitly handled:

\begin{enumerate}
\def\labelenumi{\arabic{enumi}.}
\tightlist
\item
  \textbf{Degenerate difference bodies:} excluded by the strict lower bound versus the best 30-gon in the disk.
\item
  \textbf{Perturbations leaving the original polygon problem:} prevented by Lemma 3.1.
\item
  \textbf{Adjacent moved vertices and endpoint \(z_0\):} included in Lemma 4.1; an outer incident edge still changes.
\item
  \textbf{A merely convex, locally constant perturbation:} a generic direction makes at least one norm term strictly convex.
\item
  \textbf{Use of KKT without constraint qualification:} MFCQ is constructed explicitly in Lemma 6.1.
\item
  \textbf{A hidden unsaturated radius:} contradicted quantitatively by projective angular separation.
\item
  \textbf{The earlier wrong Poincare constant:} the corrected Dirichlet denominator uses \(\sin^2(\pi/32)\), not \(\sin^2(\pi/16)\).
\item
  \textbf{Floating-point code decisions:} all are replaced by outward integer intervals.
\item
  \textbf{Objective normalization:} all gradients and Hessians here are derived from the full objective (9.1), avoiding the factor-two inconsistency in the displayed derivative formulas of the prior numerical paper.
\item
  \textbf{Local versus global fixed-code optimization:} Lemma 9.3 proves uniqueness throughout the entire high-perimeter region, not merely convergence of Newton's method near one point.
\end{enumerate}

The ancillary package includes the exact verifier, its recorded output and
hashes, and a separately organized direct scan of the normalized code space.
The feasible lower bound imports Bingane's published explicit construction;
the local verifier independently encloses its perimeter formula.  An
interval-Newton enclosure would only be needed to certify additional decimal
coordinates of the optimizer, not for the existence and uniqueness theorem.

\begin{center}\rule{0.5\linewidth}{0.5pt}\end{center}

%% file: appendix_n32.tex
\subsection{Statement}

Let \(P\) range over convex 32-gons of diameter at most one. Then the global maximum of \(\operatorname{per}(P)\) is attained by a unique congruence class. Its fully saturated difference body has sign-code orbit represented by

\begin{verbatim}
+++-+-+-++--+--+--++-+-+-+++--++
\end{verbatim}

or, in an axially symmetric representative,

\begin{verbatim}
+-++--+-+-+---++--+++-+-+-++--+-.
\end{verbatim}

The unique maximizing point is the unique high-perimeter KKT point of the fixed-code problem for this orbit. A non-rigorous high-precision Newton evaluation gives

\[
 3.1403311569546193658254013805774586723120530983\ldots,
\]

in agreement with Mulansky--Potschka. The proof itself does not depend on this decimal expansion.

Throughout, put

\[
 n=32,\qquad \theta=\frac{\pi}{32},\qquad
 U=64\sin\frac{\pi}{64},\qquad
 \varepsilon=1.35\cdot10^{-13}.
\]

Here \(U\) is Reinhardt's perimeter upper bound for a small 32-gon.

\subsection{Difference bodies and reconstruction}

For a convex polygon \(P\), let

\[
 Z=P-P.
\]

Then

\[
 Z\subseteq \overline B(0,1),\qquad
 \operatorname{per}(Z)=2\operatorname{per}(P),
\]

and \(Z\) is centrally symmetric with at most 64 edges.

\subsubsection{Lemma 2.1 (reconstruction)}

Suppose \(Z\) is a nondegenerate centrally symmetric convex 64-gon, meaning that all 64 listed points are genuine vertices. Label one half of its vertices cyclically by

\[
 z_0,z_1,\ldots,z_{31},\qquad z_{32}=-z_0,
\]

and put \(e_j=z_{j+1}-z_j\). If a code \(c_j \in\{+1,-1\}\) satisfies

\[
 \sum_{j=0}^{31}c_je_j=0,                                      \tag{2.1}
\]

then the vectors \(c_j e_j\), sorted by direction, are the edge vectors of a closed convex 32-gon \(P_c\), and

\[
 P_c-P_c=Z.
\]

\textbf{Proof.} The selected set contains exactly one vector from each antipodal pair \(\{e_j,-e_j\}\). Equation (2.1) gives closure. Sorting nonzero vectors of distinct directions gives a convex polygon with positive exterior angles. The cyclic merge of its edge list with the edge list of its negative is precisely the complete cyclic edge list of \(Z\). Convex polygons with the same cyclic edge list differ only by a translation; since both difference bodies are centered at the origin, they coincide. \hfill\(\square\)

Consequently, sufficiently small perturbations of the \(z_j\) that preserve (2.1), the disk constraints, genuine-vertex inequalities and cyclic order remain feasible for the original small-polygon problem.

Summation by parts rewrites (2.1) as

\[
 -(c_0+c_{31})z_0+\sum_{j=1}^{31}(c_{j-1}-c_j)z_j=0.             \tag{2.2}
\]

All variable coefficients therefore belong to \(\{-2,0,2\}\).

\subsubsection{Lemma 2.2 (at most one interior half-vertex)}

At a local perimeter maximum whose difference body has 64 genuine vertices, at most one of \(z_0,\ldots,z_{31}\) lies strictly inside the unit disk. If it exists, its coefficient in (2.2) is nonzero.

\textbf{Proof.} Write (2.2) as \(\sum a_j z_j=0\). If two interior vertices \(z_r,z_s\) have nonzero coefficients, use the two-sided perturbation

\[
 z_r(t)=z_r+t a_s h,\qquad z_s(t)=z_s-t a_r h.                  \tag{2.3}
\]

If an interior vertex has zero coefficient, move that vertex alone. In all cases the equality is preserved. Since the moved vertices are interior and cyclic order and the genuine-vertex inequalities are open conditions, both signs of small \(t\) are feasible. The perimeter is a sum of functions of the form \(\lVert u+tv\rVert\), hence is convex in \(t\). Choosing \(h\) outside the finite collection of lines parallel to affected edges makes at least one summand strictly convex. A two-sided local maximum is impossible. The same one-variable argument shows that a unique interior vertex cannot have zero coefficient. \hfill\(\square\)

This proof covers adjacent moved vertices and the wrap-around edge; possible cancellation on their common edge does not affect the other incident edges.

\subsection{A rigorously feasible near-regular polygon}

Use the axially symmetric code

\begin{verbatim}
c = +-++--+-+-+---++--+++-+-+-++--+-.
\end{verbatim}

It satisfies \(c_{31-j}=-c_j\). Define rational perturbations \(s_0=s_{16}=s_{32}=0\), \(s_{32-j}=-s_j\), and for \(j=1,\ldots,15\), set \(s_j=m_j/10^{20}\), where

\begin{verbatim}
m = (-84838116394748, -7759417296262, -90254740014369,
     -172750062732476, -104929943709607, -37109824686737,
     -104570029200993, -46261230637792, -102998890440414,
     -56520596873701, -100433725625308, -67650481919513,
     -34867238213719, -2083994507925, -1041997253963).
\end{verbatim}

For

\[
 \phi_j(t)=j\theta+t s_j,
\]

symmetry makes the closure residual purely imaginary. If \(a_j=c_{j-1}-c_j\), its imaginary component is

\[
 H(t)=2\sum_{j=1}^{15}a_j\sin(j\theta+t s_j)+a_{16}.             \tag{3.1}
\]

The exact interval verifier proves

\[
 H(0.99999953359)>5.53\cdot10^{-17},
\]

\[
 H(0.99999953360)<-7.84\cdot10^{-17}.                            \tag{3.2}
\]

Hence the intermediate value theorem gives a root \(t_\times\) in this interval. All angle gaps at every point of the bracket lie in \((0.09,0.107)\). At the root, Lemma 2.1 reconstructs a feasible small 32-gon, and exact interval evaluation gives

\[
 U-\operatorname{per}(P(t_\times))
 <1.336299789852041\cdot10^{-13}<\varepsilon.                   \tag{3.3}
\]

Thus every global maximizer has deficit at most \(\varepsilon\).

The interval implementation obtains a rational enclosure of \(\pi\) from Machin's formula

\[
 \pi=16\arctan\frac15-4\arctan\frac1{239}
\]

and uses alternating Taylor bounds for sine. Its certified \(\pi\) interval has width below \(6.05\times 10^{-44}\).

\subsection{Saturation for the global maximizer}

This section is independent of any sign-code enumeration.

\subsubsection{The difference body has 64 genuine vertices}

If \(Z\) had fewer than 64 vertices, central symmetry would give at most 62. The maximum perimeter of an \(m\)-gon in the unit disk is \(2m \sin(\pi/m)\), so

\[
 \operatorname{per}(P)\le62\sin\frac{\pi}{62}=U_{31}.
\]

For \(U(t)=2t \sin(\pi/(2t))\),

\[
 U'(t)=2(\sin a-a\cos a)
      =2\int_0^a u\sin u\,du
      \ge \frac{\pi^2}{6t^3},\qquad a=\frac{\pi}{2t}.
\]

Therefore

\[
 U-U_{31}\ge\frac{\pi^2}{6\cdot32^3}
 >\frac{3}{2\cdot32^3}>\varepsilon.                             \tag{4.1}
\]

This contradicts (3.3). Hence \(Z\) has exactly 64 genuine vertices.

\subsubsection{Uniform localization of its normal cones}

For a vertex \(z_j\), write

\begin{itemize}
\tightlist
\item
  \(r_j=\lVert z_j\rVert\);
\item
  \(\omega_j\) for its outward normal-cone width;
\item
  \(\eta_j\) for the cone midpoint direction;
\item
  \(\phi_j\) for the polar direction of \(z_j\);
\item
  \(\delta_j=\eta_j-\phi_j\).
\end{itemize}

Cauchy's formula is

\[
 \operatorname{per}(Z)
 =\sum_{j=0}^{63}2r_j\sin\frac{\omega_j}{2}\cos\delta_j,
 \qquad \sum_j\omega_j=2\pi.                                   \tag{4.2}
\]

Put \(\mu=\pi/32\). If one normal width is \(t\), Jensen's inequality bounds the right side by

\[
 F(t)=2\sin\frac t2
 +126\sin\frac{2\pi-t}{126}.
\]

Strong concavity of \(2 \sin(t/2)\) shows that at either \(t=\mu/2\) or \(t=3\mu/2\),

\[
 2U-F(t)\ge\frac{\mu^2\sin(\mu/4)}{16}
 >\frac{9}{32^2\cdot64\cdot16}>2\varepsilon.                   \tag{4.3}
\]

Since \(\operatorname{per}(Z)\ge 2U-2\varepsilon\), monotonicity of \(F\) on each side of \(\mu\) gives

\[
 \frac\mu2<\omega_j<\frac{3\mu}{2}.                             \tag{4.4}
\]

The nonnegative deficit decomposition from (4.2) now gives

\[
 1-r_j<64\varepsilon<\frac12,                                   \tag{4.5}
\]

and

\[
 1-\cos\delta_j<128\varepsilon.                                 \tag{4.6}
\]

If \(\lvert\delta_j\rvert\ge \mu/16=\pi/512\), then

\[
 1-\cos\delta_j\ge\frac{2\delta_j^2}{\pi^2}
 \ge\frac1{131072}>128\varepsilon,
\]

contradicting (4.6). Thus

\[
 |\delta_j|<\frac\mu{16}.                                       \tag{4.7}
\]

Consecutive normal midpoints differ by \((\omega_j+\omega_{j+1})/2>\mu/2\); hence consecutive radial directions differ by more than

\[
 \frac\mu2-2\frac\mu{16}=\frac{3\mu}{8}.                        \tag{4.8}
\]

The same lower bound holds for the projective angular distance between any two distinct half-vertices.

\subsubsection{KKT excludes the last possible interior radius}

By Lemma 2.2, suppose there is exactly one interior half-vertex \(z_r\), with \(a_r=\pm 2\). MFCQ holds explicitly: move every active circle vertex radially inward and use the free displacement of \(z_r\) to repair the two components of (2.2). The equality Jacobian has rank two because the \(z_r\) block is \(a_r I_2\).

For the objective \(\operatorname{per}(P)=\operatorname{per}(Z)/2\), the gradient with respect to a half-vertex is

\[
 g_j=2\sin\frac{\omega_j}{2}e^{i\eta_j}.                         \tag{4.9}
\]

At the interior vertex, stationarity gives \(g_r=a_r \lambda\), so

\[
 ||\lambda||=\sin\frac{\omega_r}{2}
\]

and the projective direction of \(\lambda\) is \(\eta_r\). There is another index \(j \ne r\) with \(a_j \ne 0\); otherwise (2.2) would force \(z_r=0\), impossible for a vertex of the full-dimensional difference body. This second vertex is on the unit circle. Projecting its KKT equation onto its tangent gives

\[
 |\sin(\eta_r-\phi_j)|
 =\frac{\sin(\omega_j/2)}{\sin(\omega_r/2)}|\sin\delta_j|.       \tag{4.10}
\]

By (4.4), the ratio is less than

\[
 \frac{\sin(3\mu/4)}{\sin(\mu/4)}<3.
\]

If \(\vartheta\) is the projective distance from \(\eta_r\) to \(\phi_j\), then \(\sin(\vartheta)\ge 2\vartheta/\pi\) and (4.7)--(4.10) imply

\[
 \vartheta<\frac{3\pi\mu}{32}.
\]

Consequently

\[
 d_{\mathbb{RP}^1}(\phi_r,\phi_j)
 <\frac\mu{16}+\frac{3\pi\mu}{32}
 =\frac{(2+3\pi)\mu}{32}
 <\frac{3\mu}{8},                                                \tag{4.11}
\]

using \(\pi<10/3\). This contradicts (4.8). Hence every one of the 64 difference- body vertices is on the unit circle.

\subsection{The saturated angle/code model}

Rotate so that

\[
 \phi_0=0,\qquad \phi_{32}=\pi,
\]

and put

\[
 \alpha_j=\phi_{j+1}-\phi_j>0.
\]

For a code \(c \in\{\pm 1\}^{32}\), the problem is

\[
 \max \sum_{j=0}^{31}2\sin\frac{\alpha_j}{2},                   \tag{5.1}
\]

subject to

\[
 G_c(\phi)=\sum_{j=0}^{31}c_j
 (e^{i\phi_{j+1}}-e^{i\phi_j})=0.                               \tag{5.2}
\]

Write

\[
 \phi_j=j\theta+s_j,\quad s_0=s_{32}=0,\quad
 x_j=s_{j+1}-s_j.
\]

Then \(\sum x_j=0\).

\subsubsection{Lemma 5.1 (all high-perimeter gaps are localized)}

Every feasible point with deficit at most \(\varepsilon\) satisfies

\[
 0.09<\alpha_j<0.107.                                            \tag{5.3}
\]

\textbf{Proof.} If one gap equals \(t\), Jensen bounds the other 31 gaps by making them equal. The resulting one-variable upper envelope is increasing for \(t<\theta\) and decreasing for \(t>\theta\). At \(t=0.09\), strong concavity and \(\pi>3\) give a deficit larger than

\[
 \frac14\left(0.045-\frac{0.045^3}{6}\right)
 \frac{32}{31}\left(\frac3{32}-0.09\right)^2>\varepsilon.
\]

At \(t=0.107\), use \(\theta<11/112\) and the lower bound \((3-0.107)/31\) for the common remaining gap; the analogous rational inequality is again larger than \(\varepsilon\). \hfill\(\square\)

On this box, strong concavity yields

\[
 U-\operatorname{per}(P)
 \ge \kappa ||x||_2^2,\qquad
 \kappa=\frac14\sin(0.045)
 >\frac14\left(0.045-\frac{0.045^3}{6}\right).                  \tag{5.4}
\]

The verifier checks that the right-hand lower bound times \((3.47\times 10^{-6})^{2}\) exceeds \(\varepsilon\). Therefore

\[
 ||x||_2<R_0:=3.47\cdot10^{-6}.                                 \tag{5.5}
\]

\subsection{A uniform code-screening inequality}

Let \(\xi=e^{i \theta}\). At the regular point,

\[
 b_c=(\xi-1)\sum_{j=0}^{31}c_j\xi^j.                             \tag{6.1}
\]

Taylor expansion gives

\[
 G_c(s)=b_c+L_c(s)+R_c(s),                                       \tag{6.2}
\]

where

\[
 L_c(s)=i\sum_{j=0}^{31}c_j
 (\xi^{j+1}s_{j+1}-\xi^j s_j).                                  \tag{6.3}
\]

The scalar exponential remainder and the Dirichlet Poincare inequality give

\[
 |R_c(s)|\le||s||_2^2
 \le\frac{||x||_2^2}{4\sin^2(\pi/64)}
 \le256||x||_2^2.                                                \tag{6.4}
\]

The following uniform linear estimate is useful beyond \(n=32\).

\subsubsection{Lemma 6.1 (twisted-Dirichlet operator bound)}

For general \(n\),

\[
 |L_c(s)|\le\sqrt{2n}\cos\frac{\pi}{2n}\,||x||_2.               \tag{6.5}
\]

\textbf{Proof.} Cauchy--Schwarz gives

\[
 |L_c(s)|^2\le n\sum_{j=0}^{n-1}|\xi s_{j+1}-s_j|^2.
\]

The twisted Dirichlet quadratic on \(s_1,\ldots,s_{n-1}\) has eigenvalues

\[
 \mu_k=2-2\cos\theta\cos(k\theta),
\]

whereas the ordinary Dirichlet energy \(\sum x_j^{2}\) has eigenvalues

\[
 \lambda_k=2-2\cos(k\theta).
\]

For \(k\ge 1\),

\[
 \mu_k\le2\cos^2(\theta/2)\lambda_k,
\]

because half the difference between the right and left sides is \(\cos(\theta)-\cos(k \theta)\ge 0\). This proves (6.5). \hfill\(\square\)

For \(n=32\), (6.5) is less than \(8\lVert x\rVert\). Thus every code capable of attaining deficit at most \(\varepsilon\) must satisfy

\[
 |b_c|<8R_0+256R_0^2<3.41\cdot10^{-5}.                           \tag{6.6}
\]

\subsection{Exact meet-in-the-middle code certificate}

Overall negation of a code is immaterial, so normalize \(c_0=+1\). There are \(2^{31}=2{,}147{,}483{,}648\) normalized half-codes.

The verifier encloses every coordinate of every edge

\[
 e_j=\xi^{j+1}-\xi^j
\]

by a rational interval and chooses an integer vector \(E_j\) at scale \(M=10^{18}\) such that each coordinate differs from \(E_j/M\) by less than \(1/M\). It splits the code into blocks \(0,\ldots,15\) and \(16,\ldots,31\), forms all \(2^{15}\) and \(2^{16}\) signed integer sums, sorts the second block by its first coordinate, and performs an exact integer range search. Since the total Euclidean rounding error is less than \(64/M\), every true residual satisfying (6.6) is included in the integer search circle.

The exact result is:

\[
 \boxed{96\text{ normalized codes survive}.}                    \tag{7.1}
\]

They are exactly three disjoint dihedral orbits, each of size 32, represented by

\begin{verbatim}
A = ++++--+--+-+-+-+++---+-+-+-++-++
B = +++-+-+-++--+--+--++-+-+-+++--++
C = ++-++-+-+-+-+--++--+-+-+-+-++-++.
\end{verbatim}

The axially symmetric code of Section 3 belongs to orbit \(B\).

This certificate scans every normalized code, not only axially symmetric ones. It therefore does not assume Mossinghoff's symmetry conjecture.

\subsection{Exact exclusion of orbits A and C}

For a fixed code let

\[
 \sigma_c=\sup_{s\ne0}\frac{|L_c(s)|}{||x||_2}.
\]

Writing \(a_j=c_{j-1}-c_j\), the squared norm is the largest eigenvalue of the 2-by-2 matrix

\[
 K_c=B_cD^{-1}B_c^T,                                             \tag{8.1}
\]

where \(D\) is the Dirichlet path matrix and the rows of \(B_c\) are

\[
 (-a_j\sin(j\theta))_{j=1}^{31},\qquad
 ( a_j\cos(j\theta))_{j=1}^{31}.                                \tag{8.2}
\]

The verifier uses

\[
 (D^{-1})_{jk}=\frac{\min(j,k)(32-\max(j,k))}{32}                \tag{8.3}
\]

and exact fixed-point error bounds. Gershgorin gives the certified squared- norm upper bounds

\[
 \sigma_A^2<16.36425,\qquad
 \sigma_B^2<16.45308,\qquad
 \sigma_C^2<16.25477,                                           \tag{8.4}
\]

so in particular every survivor has \(\sigma_c<4.1\).

The quantities \(\lvert b_c\rvert\) and \(\sigma_c\) are invariant under the full dihedral action and overall sign change: rotations multiply the complex residual and linear map by a unit complex number, while reflections conjugate them and reverse the Dirichlet path. It is therefore enough to check one representative of each orbit.

The same exact residual intervals prove

\[
 |b_A|>2.31\cdot10^{-5},\qquad
 |b_C|>1.66\cdot10^{-5}.                                        \tag{8.5}
\]

For any high-perimeter feasible point of either code, (6.2), (6.4), (5.5), and (8.4) imply

\[
 ||x||_2>
 \frac{|b_c|-256R_0^2}{4.1}.                                    \tag{8.6}
\]

Combining (8.6) with (5.4), the exact rational lower deficits are

\[
 U-\operatorname{per}(P)>3.5689976\cdot10^{-13}
 \quad\text{for orbit A},                                       \tag{8.7}
\]

and

\[
 U-\operatorname{per}(P)>1.8428631\cdot10^{-13}
 \quad\text{for orbit C}.                                       \tag{8.8}
\]

Both exceed \(\varepsilon\), so neither orbit can contain a global maximizer. Orbit \(B\) is the unique surviving code class.

\subsection{Continuous global uniqueness inside orbit B}

Use the axial representative of Section 3. It has 21 switch indices

\[
 J=\{j: c_{j-1}\ne c_j\}.
\]

Let \(f=-\operatorname{per}(P)\) and let \(g=(\Re G_c,\Im G_c)\).

\subsubsection{Strong convexity}

The Hessian of \(f\) is the weighted Dirichlet path matrix with edge weights

\[
 w_j=\frac12\sin\frac{\alpha_j}{2}.
\]

Using (5.3), \(\sin(\pi/64)\ge 1/32\), and the alternating lower bound for \(\sin(0.045)\),

\[
 \nabla^2 f\succeq m_0 I,
 \qquad
 m_0:=\frac{0.045-0.045^3/6}{512}
 >8.78609\cdot10^{-5}.                                          \tag{9.1}
\]

\subsubsection{Constraint qualification and multiplier bound}

At a point \(\phi\), the nonzero columns of \(Dg\) are \(a_j i e^{i\phi_j}\). Hence

\[
 \sigma_{\min}(Dg)^2
 =2\left(21-\left|\sum_{j\in J}e^{2i\phi_j}\right|\right).       \tag{9.2}
\]

The exact fixed-point certificate shows that at the regular point

\[
 \left|\sum_{j\in J}e^{2ij\theta}\right|<2.                     \tag{9.3}
\]

Furthermore, (6.4) gives \(\lVert s\rVert<16R_0\). Since \(\sqrt{21}<5\),

\[
 \left|\sum_{j\in J}
 (e^{2i\phi_j}-e^{2ij\theta})\right|
 \le2\sqrt{21}\,||s||<160R_0<1.                                 \tag{9.4}
\]

Equations (9.2)--(9.4) imply

\[
 \sigma_{\min}(Dg)>6.                                           \tag{9.5}
\]

Thus LICQ holds throughout the high-perimeter ball.

At the regular point the objective gradient vanishes. If \(r_j=\cos(\alpha_j/2)-\cos(\theta/2)\), then the path difference operator and (5.3) give

\[
 ||\nabla f||_2\le\sin(0.107/2)||x||_2
 <\frac{0.107}{2}R_0.                                           \tag{9.6}
\]

At a KKT point, \(\nabla f+(Dg)^T y=0\), so (9.5)--(9.6) give

\[
 ||y||<Y_0:=\frac{(0.107/2)R_0}{6}
 <3.095\cdot10^{-8}.                                             \tag{9.7}
\]

\subsubsection{Two KKT points are impossible}

Suppose \(u,v\) are two high-perimeter feasible KKT points and put \(d=u-v\). For the complex closure constraint, \(\lvert a_j\rvert\le 2\) and the scalar exponential remainder imply

\[
 ||g(u)-g(v)-Dg(v)d||_2\le||d||_2^2,                             \tag{9.8}
\]

and the analogous bound holds with \(u\) and \(v\) interchanged. Subtract the two stationarity equations and take the inner product with \(d\). Strong convexity, feasibility, (9.8), and (9.7) give

\[
 m_0||d||_2^2\le2Y_0||d||_2^2.                                  \tag{9.9}
\]

But the exact verifier proves

\[
 m_0-2Y_0>8.7799\cdot10^{-5}>0.                                 \tag{9.10}
\]

Therefore \(d=0\). Orbit \(B\) has at most one high-perimeter KKT point after fixing rotation.

\subsection{Completion of the global argument}

The original polygon problem has a global maximizer by compactness. Sections 3--4 show that every global maximizer has a saturated 64-vertex difference body. Sections 5--8 show that its code must lie in orbit \(B\). Its angle gaps are strictly between \(0.09\) and \(0.107\), so no order inequality is active, and Section 9 gives LICQ; hence it is a KKT point. Section 9 proves that such a point is unique for a normalized representative. Dihedral changes of code, rotation, reflection and translation produce congruent polygons.

Thus there is exactly one maximizing congruence class, and its code is orbit \(B\). This proves the theorem.

\subsection{What is and is not machine-assisted}

Analytic:

\begin{enumerate}
\def\labelenumi{\arabic{enumi}.}
\tightlist
\item
  difference-body reduction and reconstruction;
\item
  at-most-one-interior-vertex perturbation;
\item
  MFCQ/KKT saturation;
\item
  gap localization and strong-concavity radius bound;
\item
  uniform twisted-Dirichlet operator estimate;
\item
  residual elimination inequalities;
\item
  best-code KKT uniqueness.
\end{enumerate}

Finite certificate:

\begin{enumerate}
\def\labelenumi{\arabic{enumi}.}
\tightlist
\item
  the one-dimensional feasible root and its perimeter bound;
\item
  interval enclosures of roots of unity;
\item
  exact coverage of all \(2^{31}\) normalized codes by meet-in-the-middle;
\item
  the partition of 96 survivors into three dihedral orbits;
\item
  code-specific operator-norm bounds and residual lower bounds;
\item
  the regular switch-root bound used for LICQ.
\end{enumerate}

The verifier uses only the Python standard library (\texttt{Fraction}, integers, \texttt{bisect}, and elementary data structures). Its proof decisions do not use binary floating point. Decimal-looking values in the printed report are only human-readable renderings after all assertions have passed.

\subsection{Certificate audit and cross-checks}

The ancillary package includes the exact meet-in-the-middle verifier, its
recorded output and hashes, a separate normalized-code scan, an independent
orbit comparison, and a high-precision stationary-point cross-check.  These
artifacts reproduce the 96 survivors, their three dihedral orbits, and the
winning orbit used above.  An interval-Newton enclosure would only be needed
to certify additional decimal coordinates of the optimizer.

%% file: appendix_n64.tex
\subsection{Difference-body reduction and reconstruction}

Throughout this appendix, put
\[
 n=64,\qquad \theta=\frac{\pi}{64},\qquad
 U=128\sin\frac{\pi}{128},\qquad
 \varepsilon=2.84\cdot10^{-23}.
\]
Here \(U\) is Reinhardt's perimeter upper bound for a small 64-gon.

For a convex polygon \(P\), let \(Z=P-P\). Then

\[
 Z\subseteq \overline B(0,1),\qquad
 \operatorname{per}(Z)=2\operatorname{per}(P),
\]

and \(Z\) is centrally symmetric with at most 128 edges.

\subsubsection{Lemma 1.1 (reconstruction)}

Let \(Z\) be a centrally symmetric convex 128-gon with all listed vertices genuine. Label one half cyclically by

\[
 z_0,\ldots,z_{63},\qquad z_{64}=-z_0,
\]

and put \(e_j=z_{j+1}-z_j\). If a code \(c_j \in\{\pm 1\}\) satisfies

\[
 \sum_{j=0}^{63}c_je_j=0,                                      \tag{1.1}
\]

then the vectors \(c_j e_j\), sorted by direction, are the edges of a closed convex 64-gon \(P_c\), and

\[
 P_c-P_c=Z.
\]

\textbf{Proof.} The selected vectors contain exactly one vector from every antipodal pair \(\{e_j,-e_j\}\). They are nonzero and have distinct directions. Equation (1.1) gives closure. Sorting by direction therefore produces a strictly convex polygon. Merging its edge list with the edge list of its negative gives the full cyclic edge list of \(Z\); centered polygons with the same cyclic edge list coincide. \hfill\(\square\)

Summation by parts rewrites (1.1) as

\[
 -(c_0+c_{63})z_0+\sum_{j=1}^{63}(c_{j-1}-c_j)z_j=0,             \tag{1.2}
\]

whose coefficients lie in \(\{0,\pm 2\}\).

\subsubsection{Lemma 1.2 (at most one interior half-vertex)}

At a local maximum with a 128-vertex difference body, at most one of \(z_0,\ldots,z_{63}\) lies strictly inside the unit disk, and such a vertex has a nonzero coefficient in (1.2).

\textbf{Proof.} If two interior vertices have nonzero coefficients \(a_r,a_s\), use

\[
 \delta z_r=a_sh,\qquad \delta z_s=-a_rh.
\]

If an interior vertex has zero coefficient, move it alone. For sufficiently small perturbations of either sign, the disk constraints, cyclic order and genuine-vertex conditions remain valid. Along the perturbation the perimeter is a sum of norms of affine functions, hence convex. A generic \(h\) makes one affected norm strictly convex, contradicting a two-sided local maximum. \hfill\(\square\)

\subsection{A rigorously feasible near-regular 64-gon}

Use the winning axial code \(c\), satisfying

\[
 c_{63-j}=-c_j.
\]

Let \(s_0=s_{32}=s_{64}=0\), \(s_{64-j}=-s_j\), and set

\[
 s_j=m_j/10^{40}\qquad(1\le j\le31),
\]

where the integer list \(m_j\) is contained in \texttt{n64\_analytic\_verifier.py}. Put

\[
 \phi_j(t)=j\theta+t s_j.
\]

Writing \(a_j=c_{j-1}-c_j\), axial symmetry makes the closure residual purely imaginary:

\[
 G(t)=iH(t),\qquad
 H(t)=2\sum_{j=1}^{31}a_j\sin(j\theta+t s_j)+a_{32}.             \tag{2.1}
\]

Exact rational interval arithmetic proves

\[
 H(0.999999999987873350)>5.67\cdot10^{-28},
\]

\[
 H(0.999999999987873353)<-6.00\cdot10^{-28}.                    \tag{2.2}
\]

Hence an exact root \(t_\times\) exists in that bracket. At every point in the bracket all angle gaps lie in \((0.045,0.054)\). The same interval calculation gives

\[
 U-\operatorname{per}(P(t_\times))
 <2.835602\cdot10^{-23}<\varepsilon.                            \tag{2.3}
\]

Lemma 1.1 therefore reconstructs a genuinely feasible small 64-gon.

The verifier obtains \(\pi\) from Machin's formula and bounds sine by alternating Taylor intervals. Its certified \(\pi\) interval has width below \(6.05\times 10^{-44}\).

\subsection{Saturation of every global maximizer}

\subsubsection{The difference body has 128 genuine vertices}

If \(Z\) had fewer than 128 vertices, central symmetry would give at most 126, so

\[
 \operatorname{per}(P)\le126\sin\frac\pi{126}=U_{63}.
\]

For \(V(t)=2t \sin(\pi/(2t))\), the standard derivative estimate gives

\[
 U-U_{63}>\frac{3}{2\cdot64^3}>\varepsilon,
\]

contradicting (2.3). Thus \(Z\) has 128 genuine vertices.

\subsubsection{Near-regular normal cones}

For a difference-body vertex write

\[
 z_j=r_je^{i\phi_j},
\]

and let \(\omega_j\) be its normal-cone width, \(\eta_j\) its normal-cone midpoint, and \(\delta_j=\eta_j-\phi_j\). Cauchy's formula is

\[
 \operatorname{per}(Z)=
 \sum_{j=0}^{127}2r_j\sin\frac{\omega_j}{2}\cos\delta_j,
 \qquad \sum_j\omega_j=2\pi.                                  \tag{3.1}
\]

Put \(\mu=\pi/64\). Fixing one width and averaging the other 127 widths with Jensen's inequality, together with strong concavity, shows that a width outside

\[
 \frac\mu2<\omega_j<\frac{3\mu}{2}                              \tag{3.2}
\]

would cost more than \(2\varepsilon\). The verifier checks both Jensen boundary values and, in particular, the coarse lower bound

\[
 \frac{9}{64^2\cdot128\cdot16}>2\varepsilon.
\]

The nonnegative deficit decomposition in (3.1) then yields

\[
 1-r_j<128\varepsilon,
 \qquad 1-\cos\delta_j<256\varepsilon.                          \tag{3.3}
\]

Consequently

\[
 |\delta_j|<\frac\mu{16}.                                      \tag{3.4}
\]

Successive radial directions, and hence any two different half-vertices in projective angle, are separated by more than

\[
 \frac\mu2-2\frac\mu{16}=\frac{3\mu}{8}.                       \tag{3.5}
\]

\subsubsection{KKT excludes the last interior radius}

By Lemma 1.2 suppose only \(z_r\) is interior. Its coefficient is \(a_r=\pm 2\). MFCQ holds: move active circle vertices radially inward and use the free 2-dimensional displacement of \(z_r\) to repair (1.2).

For the objective \(\operatorname{per}(Z)/2\), the half-vertex gradient is

\[
 g_j=2\sin\frac{\omega_j}{2}e^{i\eta_j}.                        \tag{3.6}
\]

At \(r\), stationarity gives \(g_r=a_r \lambda\), so the projective direction of \(\lambda\) is \(\eta_r\). There is another \(j \ne r\) with \(a_j \ne 0\); otherwise (1.2) would force \(z_r=0\). This second vertex is on the unit circle. Project its KKT equation onto its circle tangent to obtain

\[
 |\sin(\eta_r-\phi_j)|
 =\frac{\sin(\omega_j/2)}{\sin(\omega_r/2)}|\sin\delta_j|.      \tag{3.7}
\]

By (3.2) the ratio is less than three. Equations (3.4), (3.7) and \(\sin t\ge 2t/\pi\) give

\[
 d_{\mathbb {RP}^1}(\phi_r,\phi_j)
 <\frac{(2+3\pi)\mu}{32}<\frac{3\mu}{8},
\]

contradicting (3.5). Thus all 128 difference-body vertices lie on the unit circle. This completes the saturation argument for \(n=64\).

\subsection{Saturated angle model and localization}

Rotate so that \(\phi_0=0\), \(\phi_{64}=\pi\), and put

\[
 \alpha_j=\phi_{j+1}-\phi_j,
 \qquad
 \phi_j=j\theta+s_j,
 \qquad x_j=s_{j+1}-s_j.
\]

The objective is

\[
 F=\sum_{j=0}^{63}2\sin\frac{\alpha_j}{2},                     \tag{4.1}
\]

and closure is

\[
 G_c=\sum_{j=0}^{63}c_j(e^{i\phi_{j+1}}-e^{i\phi_j})=0.         \tag{4.2}
\]

A one-gap Jensen bound, with both boundary values checked by the core analytic verifier, proves that every feasible point with deficit at most \(\varepsilon\) satisfies

\[
 0.045<\alpha_j<0.054.                                          \tag{4.3}
\]

On this box, strong concavity gives

\[
 U-F\ge \kappa\|x\|_2^2,
 \qquad
 \kappa>\frac14\left(0.0225-\frac{0.0225^3}{6}\right).
\]

The exact verifier checks

\[
 \boxed{\|x\|_2<R_0:=7.11\cdot10^{-11}.}                       \tag{4.4}
\]

\subsection{Uniform regular-residual screen}

Let \(\xi=e^{i \theta}\) and

\[
 b_c=(\xi-1)\sum_{j=0}^{63}c_j\xi^j.                            \tag{5.1}
\]

Taylor expansion gives

\[
 G_c=b_c+L_c(s)+R_c(s).
\]

The Dirichlet Poincare inequality and \(\sin(\pi/128)>1/45\) give

\[
 |R_c(s)|\le512\|x\|_2^2.                                      \tag{5.2}
\]

The twisted-Dirichlet operator estimate, valid for general \(n\), is

\[
 |L_c(s)|\le\sqrt{2n}\cos\frac\pi{2n}\,\|x\|_2.               \tag{5.3}
\]

For \(n=64\), the coefficient is less than \(23/2\). Thus every competitive code satisfies

\[
 |b_c|<8.2\cdot10^{-10}.                                       \tag{5.4}
\]

Since \(\lvert \xi-1\rvert=2\sin(\pi/128)>0.049\), it must satisfy

\[
 \left|\sum_{j=0}^{63}c_j\xi^j\right|<1.7\cdot10^{-8}.          \tag{5.5}
\]

\subsection{Orthogonal reflection-pair decomposition}

This is the new device that makes a full \(n=64\) code certificate feasible. Pair indices \(j\) and \(63-j\), for \(0\le j\le 31\), and put

\[
 \beta_j=\left(j-\frac{63}{2}\right)\theta,
\]

\[
 p_j=\frac{c_j+c_{63-j}}2,
 \qquad
 q_j=\frac{c_j-c_{63-j}}2.
\]

For each \(j\), exactly one of \(p_j,q_j\) is zero and the other belongs to \(\{\pm 1\}\). Direct calculation gives

\[
 \sum_{j=0}^{63}c_j\xi^j
 =2e^{63i\theta/2}
 \left(\sum_{j=0}^{31}p_j\cos\beta_j
       +i\sum_{j=0}^{31}q_j\sin\beta_j\right).                 \tag{6.1}
\]

Thus the real and imaginary coordinates after a fixed rotation are two orthogonal one-dimensional ternary subset sums, with complementary supports. Conversely every pair of ternary coefficient vectors with complementary supports determines exactly one 64-bit code. Hence (6.1) is a bijective reparameterization of all

\[
 4^{32}=2^{64}
\]

half-codes, not a symmetry restriction.

\subsection{Exact exhaustive code certificate}

At scale \(M=10^{16}\), the core analytic verifier proves that each of the 32 horizontal and 32 vertical weights in (6.1) lies within \(1/M\) of the integer printed in \texttt{n64\_code\_fixed.cpp}. Split each ternary sum into two blocks of 16 terms. Each block has only

\[
 3^{16}=43,046,721
\]

states. Sorting the two block lists and using a moving interval window enumerates all coordinate sums within the conservative fixed-point radius. Horizontal and vertical records are then matched exactly by complementary support masks.

If the true complex sum satisfies (5.5), its fixed-point coordinate vector is inside the integer circle used by the program; the accumulated rounding error is less than 64 integer units. Therefore the scan cannot omit a competitive code.

The exact integer result is

\[
 \boxed{896\text{ half-codes survive}.}                         \tag{7.1}
\]

They form exactly six full-dihedral orbits, with orbit sizes

\[
 128,128,128,128,128,256.                                      \tag{7.2}
\]

Five orbits contain reflection-symmetric representatives; the last orbit is generic. The winning published axial code is the first orbit.

The recorded packaging run took 20.43 seconds and 1.47 GB; runtime is environment dependent and is not a proof assertion. The decomposition replaces an infeasible direct scan of \(2^{63}\) normalized codes.

\subsection{Exact elimination of the five nonwinning orbits}

For a fixed code define

\[
 \sigma_c=\sup_{s\ne0}\frac{|L_c(s)|}{\|x\|_2}.
\]

As in the \(n=32\) proof,

\[
 \sigma_c^2=\lambda_{\max}(B_cD^{-1}B_c^T),                    \tag{8.1}
\]

where \(D\) is the 63-by-63 Dirichlet path matrix and

\[
 (D^{-1})_{jk}=\frac{\min(j,k)(64-\max(j,k))}{64}.              \tag{8.2}
\]

Exact fixed-point Gershgorin bounds prove for all six representatives

\[
 \sigma_c^2<36.                                                 \tag{8.3}
\]

For the five nonwinning representatives, exact regular-residual lower bounds are respectively

\[
 |b_c|>
 4.86,\ 5.08,\ 6.28,\ 7.64,\ 8.23
 \quad\text{times }10^{-10}.                                   \tag{8.4}
\]

For these residual comparisons the core post-verifier uses a 256-unit fixed-point padding, larger than the generic Euclidean rounding bound \(128\sqrt2\). If one of these codes had deficit at most \(\varepsilon\), equations (4.4), (5.2), (8.3) would imply

\[
 \|x\|_2>
 \frac{|b_c|-512R_0^2}{6}.                                     \tag{8.5}
\]

Combining (8.5) with strong concavity gives deficits exceeding \(\varepsilon\); the smallest certified excess is

\[
 8.50\cdot10^{-24}.                                             \tag{8.6}
\]

Therefore only the winning orbit can contain a global maximizer.

\subsection{Continuous global uniqueness inside the winning orbit}

Use the axial representative from Section 2. It has 27 switch indices

\[
 J=\{j:c_{j-1}\ne c_j\}.
\]

Let \(f=-F\) and \(g=(\Re G_c,\Im G_c)\).

\subsubsection{Strong convexity}

The Hessian of \(f\) is the weighted Dirichlet path matrix with weights

\[
 w_j=\frac12\sin\frac{\alpha_j}{2}.
\]

Using (4.3) and \(\sin(\pi/128)>1/45\), throughout the high-perimeter ball

\[
 \nabla^2f\succeq m_0I,
 \qquad
 m_0>
 \frac{2}{45^2}
 \left(0.0225-\frac{0.0225^3}{6}\right).                       \tag{9.1}
\]

\subsubsection{Constraint Jacobian and multiplier bound}

The nonzero columns of \(Dg\) are \(a_j i e^{i\phi_j}\). Therefore

\[
 \sigma_{\min}(Dg)^2
 =2\left(27-\left|\sum_{j\in J}e^{2i\phi_j}\right|\right).     \tag{9.2}
\]

At the regular point the exact certificate proves

\[
 \left|\sum_{j\in J}e^{2ij\theta}\right|<5.                    \tag{9.3}
\]

Poincare gives \(\lVert s\rVert<23R_0\), so the root sum changes by less than \(276R_0<1\). Hence its norm remains below six and

\[
 \sigma_{\min}(Dg)>6.                                          \tag{9.4}
\]

The objective gradient obeys

\[
 \|\nabla f\|_2
 \le\sin(0.054/2)\|x\|_2
 <(0.054/2)R_0.
\]

At a KKT point,

\[
 \|y\|<Y_0:=\frac{(0.054/2)R_0}{6}.                            \tag{9.5}
\]

\subsubsection{Two KKT points cannot coexist}

If \(u,v\) are feasible high-perimeter KKT points and \(d=u-v\), strong convexity gives

\[
 d^T(\nabla f(u)-\nabla f(v))\ge m_0\|d\|^2.                   \tag{9.6}
\]

The scalar exponential remainder and \(\lvert a_j\rvert\le 2\) give

\[
 \|g(u)-g(v)-Dg(v)d\|\le\|d\|^2,                               \tag{9.7}
\]

and the analogous reverse estimate. Subtracting the two stationarity equations and using feasibility yields

\[
 m_0\|d\|^2\le2Y_0\|d\|^2.                                   \tag{9.8}
\]

The exact verifier proves

\[
 m_0-2Y_0>2.22\cdot10^{-5}>0.                                 \tag{9.9}
\]

Thus \(d=0\). The winning orbit has at most one high-perimeter KKT point after fixing rotation.

\subsection{Completion}

The original problem has a global maximizer by compactness. Sections 2--3 show every global maximizer has a saturated 128-vertex difference body. Sections 4--8 prove that its code belongs to the winning orbit. Its gaps are strictly inside (4.3), and (9.4) gives LICQ, so it is a KKT point. Section 9 proves that this point is unique after normalization. Dihedral code changes, rotation, reflection and translation give congruent polygons.

Therefore there is exactly one maximizing congruence class of small 64-gons.

\subsection{Machine-assisted components}

Analytic:

\begin{enumerate}
\def\labelenumi{\arabic{enumi}.}
\tightlist
\item
  difference-body reconstruction;
\item
  at-most-one-interior-vertex perturbation;
\item
  saturation via normal-cone localization and KKT;
\item
  high-perimeter angle localization;
\item
  uniform regular-residual screen;
\item
  reflection-pair orthogonal decomposition;
\item
  nonbest-orbit lower-deficit inequalities;
\item
  best-orbit KKT uniqueness.
\end{enumerate}

Finite certificate:

\begin{enumerate}
\def\labelenumi{\arabic{enumi}.}
\tightlist
\item
  rigorous one-dimensional feasible root and deficit;
\item
  root-of-unity fixed-point enclosures;
\item
  exhaustive coverage of all \(2^{64}\) codes via two ternary MITM scans;
\item
  partition of 896 survivors into six dihedral orbits;
\item
  residual lower bounds and operator-norm upper bounds;
\item
  switch-root bound for continuous uniqueness.
\end{enumerate}

\subsection{Certificate audit and cross-checks}

The ancillary package includes the analytic verifier, the C++17 exhaustive
scan, the post-verifier, their recorded outputs and hashes, and separately
organized survivor, high-precision KKT, and orthogonal-pair cross-checks.
Together they reproduce the 896 survivors, their six dihedral orbits, and the
winning orbit used above.  An interval-Newton enclosure would only be needed
to certify additional decimal coordinates of the optimizer.

%% file: main.bbl
\begin{thebibliography}{10}

\bibitem{reinhardt1922}
K.~Reinhardt.
\newblock Extremale Polygone gegebenen Durchmessers.
\newblock \emph{Jahresbericht der Deutschen Mathematiker-Vereinigung}, 31:251--270, 1922.

\bibitem{mossinghoff2006}
M.~J. Mossinghoff.
\newblock A \$1 problem.
\newblock \emph{American Mathematical Monthly}, 113(5):385--402, 2006.
\newblock doi:10.1080/00029890.2006.11920320.

\bibitem{mossinghoff2008}
M.~J. Mossinghoff.
\newblock An isodiametric problem for equilateral polygons.
\newblock In T.~Amdeberhan and V.~H. Moll, editors, \emph{Tapas in Experimental Mathematics}, volume 457 of \emph{Contemporary Mathematics}, pages 237--252. American Mathematical Society, Providence, RI, 2008.
\newblock doi:10.1090/conm/457/08913.

\bibitem{haremossinghoff2013}
K.~G. Hare and M.~J. Mossinghoff.
\newblock Sporadic Reinhardt polygons.
\newblock \emph{Discrete \& Computational Geometry}, 49(3):540--557, 2013.
\newblock doi:10.1007/s00454-012-9479-4.

\bibitem{haremossinghoff2019}
K.~G. Hare and M.~J. Mossinghoff.
\newblock Most Reinhardt polygons are sporadic.
\newblock \emph{Geometriae Dedicata}, 198:1--18, 2019.
\newblock doi:10.1007/s10711-018-0326-5.

\bibitem{bingane2022}
C.~Bingane.
\newblock Tight bounds on the maximal perimeter and the maximal width of convex small polygons.
\newblock \emph{Journal of Global Optimization}, 84(4):1033--1051, 2022.
\newblock doi:10.1007/s10898-022-01181-9.

\bibitem{binganeaudet2022}
C.~Bingane and C.~Audet.
\newblock Tight bounds on the maximal perimeter of convex equilateral small polygons.
\newblock \emph{Archiv der Mathematik}, 119(3):325--336, 2022.
\newblock doi:10.1007/s00013-022-01745-x.

\bibitem{bingane2023}
C.~Bingane.
\newblock Maximal perimeter and maximal width of a convex small polygon.
\newblock arXiv:2106.11831, version 2, 2023.

\bibitem{mulanskypotschka2025}
B.~Mulansky and A.~Potschka.
\newblock A zonogon approach for computing small convex polygons of maximum perimeter.
\newblock \emph{Mathematical Programming}, 2025.
\newblock doi:10.1007/s10107-025-02244-x.

\bibitem{mulanskypotschka2025correction}
B.~Mulansky and A.~Potschka.
\newblock Correction: A zonogon approach for computing small convex polygons of maximum perimeter.
\newblock \emph{Mathematical Programming}, 2025.
\newblock doi:10.1007/s10107-025-02257-6.

\end{thebibliography}
